\documentclass[12pt]{article}

\usepackage{amsfonts,amsmath, amssymb}
\usepackage{dutchcal}

\usepackage{xcolor}

\usepackage{enumerate}
\usepackage{epsf,epsfig,amsfonts,a4wide}
\usepackage{amsfonts, mathdots}
\usepackage{amsmath,amssymb,amsthm}
\usepackage{url}
\usepackage{amsmath,amssymb,amsthm}

\usepackage{amssymb}
\usepackage{amsthm}
\usepackage{epsf,epsfig,amsfonts,graphicx}

\newtheorem{Theorem}{Theorem}[section]
\newtheorem{Lemma}{Lemma}[section]
\newtheorem{Proposition}{Proposition}[section]
\newtheorem{Corollary}{Corollary}[section]

\newtheorem{Remark}{Remark}[section]

\theoremstyle{remark}

\newcommand{\be}{\begin{equation}}
\newcommand{\ee}{\end{equation}}

\newcommand{\Ph}{\mathcal P}

\newcommand{\ddd}{\mathrm{d}}

\newcommand{\pd}[2]{\frac{\partial#1}{\partial#2}}
\newcommand{\vpd}[2]{\frac{\delta#1}{\delta#2}}

\newcommand{\weg}[1]{}

\title{Geometry of inhomogeneous Poisson brackets, multicomponent Harry Dym hierarchies and multicomponent Hunter-Saxton equations}
\author{Andrey Yu. Konyaev}
\date{November 2021}

\begin{document}

\maketitle

\begin{abstract}
We introduce a natural class of multicomponent local Poisson structures $\mathcal P_k + \mathcal P_1$, where $\mathcal P_1$ is local Poisson bracket of order one and $\mathcal P_k$ is a homogeneous Poisson bracket of odd order $k$ under assumption that is has Darboux coordinates (Darboux-Poisson bracket) and non-degenerate. For such brackets we obtain general formulas in arbitrary coordinates, find normal forms (related to Frobenius triples) and provide the description of the Casimirs, using purely algebraic procedure. In two-component case we completely classify such brackets up to the point transformation. From the description of Casimirs we derive new Harry Dym (HD) hierarchies and new Hunter-Saxton (HS) equations for arbitrary number of components. In two component case our HS equation differs from the well-known HS2 equation.

{\bf MSC classes:}  37K05, 37K10,  37K25,  53B20,  53B30, 53B50, 35Q53, 35Q35
\end{abstract}

\section{Introduction}

We start with a motivational example. Recall that Harry Dym (HD) equation is a non-linear partial differential equation on a function $u(t, x)$, usually written in compact form
\begin{equation}\label{hd0}
    u_t = \Bigg(\frac{1}{\sqrt{2u}}\Bigg)_{xxx}
\end{equation}
It was discovered by Kruscal in an unpublished work of Harry Dym (see book \cite{kr}) and rediscovered by Sabatier in his series of papers dedicated to the classical strings \cite{sab}. It plays an important role in study of infinite dimensional integrable systems and bares a deep connection with KdV, Camassa-Holm, KP and other integrable hierarchies (see \cite{bhd}, \cite{ped2} for details and further references). 

Harry Dym equation possesses bihamiltonian structure. Consider the space of real or complex-valued $C^{\infty}$ functions on $S^1$. We denote the coordinate on the circle as $x \in [0, 2\pi]$ and the function itself as $u$, omitting the dependence in $x$. The $j-$th derivative of $u$ in $x$ is denoted as $u_{x^j}$.

We are interested in the functionals in the form $\mathcal H = \int^{2 \pi}_{0} \mathcal h \ddd x$, where density $\mathcal h$ is a finite sum of terms
\begin{equation}\label{mon}
    \mathcal h_{j_1 \dots j_s} u_{x^{j_1}} \dots u_{x^{j_s}}, \quad j_1 + \dots + j_s = k.
\end{equation}
Here $\mathcal h_{j_1 \dots j_s}$ are smooth (may be complex-valued) functions of $u$ and, possibly, $x$. The corresponding densities are called homogeneous differential polynomials of degree $k$. 

Consider a pair of Poisson brackets, which act on functionals $\mathcal H_1 = \int_{0}^{2\pi} \mathcal h_1 \ddd x$ and $\mathcal H_2 = \int_{0}^{2\pi} \mathcal h_2 \ddd x$ as
\begin{equation}\label{bra}
\{\mathcal H_1, \mathcal H_2\}_1 = \int \limits_{0}^{2 \pi} \vpd{\mathcal H_1}{u} D^3 \Big( \vpd{\mathcal H_2}{u} \Big) \ddd x, \quad \{\mathcal H_1, \mathcal H_2\}_2 = \int \limits_{0}^{2 \pi} \Big( 2u \vpd{\mathcal H_1}{u} D\Big( \vpd{\mathcal H_2}{u} \Big) + u_x  \vpd{\mathcal H_1}{u} \vpd{\mathcal H_2}{u}\Big) \ddd x.    
\end{equation}
Here $D$ is total derivative in variable $x$ and $\vpd{}{u}$ is variational derivative in field variable $u$. Brackets \eqref{bra} are compatible in a usual sense: their linear combination with arbitrary constant coefficients is again a Poisson bracket. 

The equation \eqref{hd0} is Hamiltonian with respect to both brackets. Moreover, there exists a sequence of functionals $\mathcal H_i = \int_{0}^{2 \pi} \mathcal h_i \ddd x, i \geq 1$, such that the corresponding flows have the form
\begin{equation}\label{hd1}
\begin{aligned}
u_{t_0} & = 2u D \Big( \vpd{\mathcal H_{1}}{u}\Big) + u_x \vpd{\mathcal H_{1}}{u} = 0, \\
u_{t_i} & = 2u D \Big( \vpd{\mathcal H_{i + 1}}{u}\Big) + u_x \vpd{\mathcal H_{i + 1}}{u} = D^3 \Big( \vpd{\mathcal H_i}{u}\Big), \quad i \geq 1. 
\end{aligned}
\end{equation}
The flows \eqref{hd1} pairwise commute and Harry Dym equation \eqref{hd0} is $u_{t_1}$ for the functional $\mathcal H_1 = \int^{2\pi}_{0} \sqrt{2u} \ddd x$. Functionals $\mathcal H_i, i \geq 1$ are common conservation laws for all the equations \eqref{hd1}. Usually the term hierarchy refers to the collection of $u_{t_i}$, but we use it for the collection of $\mathcal H_i$ as well. 

The densities $\mathcal h_i$ for functionals $\mathcal H_i$ can be found by purely algebraic means, without solving any differential equations.
Consider the identity
\begin{equation}\label{ric}
    u = \frac{1}{2} \mathrm v^2 + \mathrm v_x.
\end{equation}
We search the solution in a form of a series
\begin{equation}\label{series}
\mathrm v = \mathcal v_1 + \mathcal v_2 + \mathcal v_3 + \dots,     
\end{equation}
where $\mathrm v_i$ is either zero or a homogeneous differential polynomial of degree $i - 1$. Substituting $\mathrm v$ into \eqref{ric} yields the recursion formulas for the components of the series
$$
\begin{aligned}
    \mathcal v^2_1 & = 2 u, \\
    \mathcal v_{i + 1} & = - \frac{1}{\mathcal v_1} \Bigg(\,\frac{1}{2} \sum \limits_{j = 2}^{i} \mathcal v_j \mathcal v_{i + 2 - j} + (\mathcal v_{i})_x\Bigg), \quad i \geq 1.
\end{aligned}
$$
The series is uniquely defined by the solution of the first equation, that is the choice of either $\mathcal v_1 = \sqrt{2u}$ or $\mathcal v_1 = - \sqrt{2u}$. Taking $\mathcal h_i = \mathcal v_{2i - 1}$ we get a collection of functionals $\mathcal H_i$. This is exactly the Harry Dym hierarchy.

We can try and continue hierarchy for negative $i$. For $i = -1$ and $i = -2$ we get 
\begin{equation}\label{neg}
    \begin{aligned}
    u_{t_0} & = D^3\Big( \vpd{\bar{\mathcal H}_{0}}{u} \Big) = 0, \\
    u_{t_{-1}} & = D^3 \Big( \vpd{\bar{\mathcal H}_{-1}}{u} \Big) = 2u D \Big( \vpd{\bar{\mathcal H}_{0}}{u} \Big) + u_x \vpd{\bar{\mathcal H}_{0}}{u}, \\
    u_{t_{-2}} & = 2u D \Big( \vpd{\bar{\mathcal H}_{-1}}{u} \Big) + u_x \vpd{\bar{\mathcal H}_{-1}}{u}
    \end{aligned}
\end{equation}
We take $\bar{\mathcal H}_0 = \int^{2 \pi}_{0} u \ddd x$. Introducing new variable $u = v_{xx}$ we get $\vpd{\bar{\mathcal H}_{-1}}{u} = v$. After renaming $t_{-2}$ as $t$ the third equation in \eqref{neg} takes form $v_{xxt} = 2 v_{xx} v_{x} + v_{xxx} v$. Integrating it by $x$ and taking the integration constant to be zero, we arrive to the Hunter-Saxton equation
\begin{equation}\label{hs}
v_{xt} = v_{xx} v + \frac{1}{2} v_x^2.     
\end{equation}
This equation originally appeared in the study of nematic liquid crystals.

All the mentioned well-known results admit a natural interpretation in terms of inhomogeneous Poisson brackets. Consider a sum of the brackets \eqref{bra}, which is again a Poisson bracket (actually, it is one of KdV brackets)
\begin{equation}\label{bra2}
\begin{aligned}
 \{\mathcal H_1, \mathcal H_2\} & = \{\mathcal H_1, \mathcal H_2\}_1 + \{\mathcal H_1, \mathcal H_2\}_2 = \\
 & = \int \limits_{0}^{2 \pi} \Bigg(\vpd{\mathcal H_1}{u} D^3 \Big(\vpd{\mathcal H_2}{u}\Big) +  2u \vpd{\mathcal H_1}{u} D \Big(\vpd{\mathcal H_2}{u}\Big) + u_x\vpd{\mathcal H_1}{u} \vpd{\mathcal H_2}{u} \Bigg) \ddd x.  
\end{aligned}
\end{equation}
Using Lemma \ref{lemm5} from section \ref{proof4} one can show that that Poisson bracket \eqref{bra2} has no Casimirs in the category of differential polynomials, except constants. 

At the same time one can extend the Poisson bracket \eqref{bra2} onto the space of the series of differential polynomials in the form \eqref{series}. We treat such series as density, which depends on all derivatives $u_{x^j}$. The corresponding functional is  
$$
\mathrm W =  \int \limits_{0}^{2 \pi} \mathrm w \ddd x = \sum \limits_{i = 1}^{\infty} \int \limits_{0}^{2 \pi} \mathcal w_i \ddd x = \sum \limits_{i = 1}^{\infty} \mathcal W_i.
$$
For Harry Dym hierarchy $\mathcal H_i = \int^{2 \pi}_{0} \mathcal h_i \ddd x$ consider the functional $\mathrm H = \int_{0}^{2 \pi} \mathrm h \ddd x$ for density $\mathrm h = \sum_{i = 1} (-1)^{i + 1} \mathcal h_i$. Then for arbitrary functional $\mathrm W = \int_{0}^{2 \pi} \mathrm w \ddd x$ the conditions \eqref{hd1} yield
$$
\{\mathrm W, \mathrm H\} = 0.
$$
In other words, the Harry Dym hierarchy defines the Casimirs of the inhomogeneous Poisson bracket. This implies, that the geometry of such bracket (the description of bracket and Casimirs) contains all the necessary information, that allows one to write the equations in any given coordinates.

We apply this idea in case of arbitrary $n$. To do that we introduce a natural class of inhomogeneous Poisson brackets, for which the geometry can be completely described. This, in turn, leads to the natural (and new!) generalizations of Harry Dym hierarchies and Hunter-Saxton equations. The brackets itself are described by a purely geometric object: a pair of flat metrics with certain compatibility conditions. The Casimirs and corresponding hierarchies are found by purely algebraic means. In our work we use the language of Dubrovin-Novikov brackets as the most general one. The further results, associated with compatible inhomogeneous brackets of the given type are given in \cite{NijApp3}.

First, we recall some basic definitions and results of the theory of Dubrovin-Novikov brackets (for more detailed exposure on the subject we refer the reader to \cite{doyle}, \cite{mokb}, \cite{gd}, \cite{olver}).

Let $\mathrm M^n$ be the $n-$dimensional ball. Denote $\Omega \mathrm M^n$ to be a space of maps from $S^1 \to \mathrm M^n$. The functionals $\mathcal H: \Omega \mathrm M^n \to \mathbb K$ (for $\mathbb K = \mathbb R$ or $\mathbb C$), we are interested in, are written (same as for $n = 1$) in the form $\mathcal H = \int^{2 \pi}_{0} \mathcal h \ddd x$. Here $\mathcal h$ is a differential polynomial in $n$ field variables.

The local homogeneous Poisson bracket of order $k$ is a bilinear skew-symmetric operation, which acts on functional $\mathcal H_1 = \int_{0}^{2 \pi} \mathcal h_1 \ddd x$ and $\mathcal H_2 = \int_{0}^{2 \pi} \mathcal h_2 \ddd x$ as
\begin{equation}\label{poisson}
    \{\mathcal H_1, \mathcal H_2\} = \int \limits^{2 \pi}_{0} \vpd{\mathcal H_1}{u^{\alpha}} h^{\alpha \beta} D^k \Big(\vpd{\mathcal H_2}{u^{\beta}} \Big) \ddd x + \int \limits^{\infty}_{0} \vpd{\mathcal H_1}{u^{\alpha}} \sum \limits_{j = 1}^k \mathcal P^{\alpha \beta}_{\mathrm{j}} D^{k - j} \Big(\vpd{\mathcal H_2}{u^{\beta}}\Big) \ddd x
\end{equation}
and satisfies the Jacobi identity. The entries of matrix $h^{\alpha \beta}$ are the functions of field variables $u^\alpha$ only. The entries of matrices $\mathcal P^{\alpha \beta}_{\mathrm j}$ are homogeneous differential polynomials of degree $j$. The skew-symmetry condition and Jacobi identity provide differential conditions on components of $\mathcal P^{\alpha \beta}_{\mathrm j}$. 

The symbolic expression
\begin{equation}\label{hamiltonian}
    \mathcal P^{\alpha \beta} = h^{\alpha \beta} D^k + \sum \limits_{j = 1}^k \mathcal P^{\alpha \beta}_{\mathrm{j}} D^{k - j}
\end{equation}
is called a Hamiltonian operator. For Poisson bracket \eqref{poisson} it plays the same role as Poisson tensor plays for finite-dimensional Poisson bracket.

For a given functional $\mathcal H = \int^{2\pi}_0 \mathcal h \ddd x$ the corresponding Hamiltonian flow is defined as 
$$
u^{\alpha}_t = \mathcal P^{\alpha q} \vpd{\mathcal H}{u^q} = h^{\alpha q} D^k \Big(\vpd{\mathcal H}{u^q}\Big) + \sum \limits_{j = 1}^k \mathcal P^{\alpha q}_{\mathrm{j}} D^{k - j} \Big(\vpd{\mathcal H}{u^q}\Big).
$$
For arbitrary test functional $\mathcal H_1 = \int^{2\pi}_0 \mathcal h_1 \ddd x$ it satisfies the obvious relation $\partial_t \mathcal H_1 = \{\mathcal H_1, \mathcal H\}$.

Under the coordinate change $\bar u(u) = (\bar u^1(u^1, \dots, u^n), \dots, \bar u^n(u^1, \dots, u^n))$ the components of the Hamiltonian operator are transformed as (formula 3.2 \cite{doyle}):
\begin{equation*} \label{coordinates}
    \bar h^{\alpha \beta} = h^{pq} \pd{\bar u^{\alpha}}{u^p}\pd{\bar u^{\beta}}{u^q}, \quad \bar{\mathcal P}_{\mathrm j}^{\alpha \beta} = \sum \limits_{i = j}^k \binom{j}{i}
    \mathcal P_{\mathrm i}^{pq} \pd{\bar u^{\alpha}}{u^p} D^{j - i}\Big[ \pd{\bar u^{\beta}}{u^q}\Big].
\end{equation*}
Here $\binom{j}{i}$ are binomial coefficients, $\bar h^{\alpha \beta}$ and $\bar{\mathcal P}_{\mathrm j}^{\alpha \beta}$ stand for the coefficients of Hamiltonian operator in coordinates $\bar u^1, \dots, \bar u^n$.

As it was mentioned earlier, the term $h^{\alpha \beta}$ does not depend on differential coordinates and under the coordinate transformation behaves like a contrvariant tensor of type $(2, 0)$. Due to skew-symmetry condition of the Poisson bracket (formula 3.3, \cite{doyle}) this tensor is symmetric for odd $k$ and skew-symmetric for even $k$. The Hamiltonian operator is called non-degenerate if $\operatorname{det} h \neq 0$.

The analog of \eqref{bra2} is an inhomogeneous Poisson bracket, for which its Hamiltonian operator is a sum of two homogeneous parts $\mathcal A$ and $\mathcal B$ of orders $1$ and $3$ respectively. Due to Theorem 7.1 in \cite{doyle} both $\mathcal A$ and $\mathcal B$ are Hamiltonian operators on their own. Thus, by definition, we have a pair of compatible Poisson brackets.

The local homogeneous Poisson bracket of order $k$ is Darboux-Poisson if there exist coordinates $u^1, \dots, u^n$, such that the Poisson bracket takes form 
$$
\{\mathcal H_1, \mathcal H_2\} = \int_{0}^{\infty} \vpd{\mathcal H_1}{u^\alpha} h^{\alpha \beta} D^k \Big(\vpd{\mathcal H_2}{u^\beta}\Big) \ddd x.
$$
Following \cite{doyle} we call such coordinates Darboux coordinates. The geometry of homogeneous Poisson brackets of order $k = 3$ was investigated by a number of authors (\cite{doyle}, \cite{pot}, \cite{pot2}, \cite{fpv}). In our work we assume that the bracket of order $k > 1$ is Darboux-Poisson.

The non-degenerate homogeneous Poisson brackets of order one are called Poisson brackets of hydrodynamic type. The physicists used this type of brackets for a long time: Landau, for example, used them without any specific name in formulation of his theory of the superfluidity for hellium (formula 1.4 in \cite{landau}). 

The geometric theory of these objects was developed by Dubrovin and Novikov in the 1980s (\cite{dn}, \cite{dn2}, see also \cite{mok1}). They showed, in particular, that these Poisson brackets always have Darboux coordinates and, thus, are in one-to-one correspondence with flat metrics. We will denote the Hamiltonian operator of non-degenerate homogeneous Poisson bracket of order one as $\mathcal A_g$, where $g$ is the corresponding flat contrvariant metric.

The paper is organized as follows: in the next section \ref{main} we formulate the main results and provide important examples. The Theorem \ref{t1} provides the formula for Darboux-Poisson brackets in local coordinates for arbitrary $k$. It generalizes the famous formula by Dobrovin and Novikov. In particular, it implies that for odd $k$ such operator is uniquely defined by a flat metric. 

Theorem \ref{t2} yields normal form for inhomogeneous Poisson brackets with the odd term of order $k \geq 3$ being homogeneous Darboux-Poisson bracket. For $k = 3$ the statement of the Theorem \ref{t2} was independently obtained by P. Lorenzoni and R. Vitolo in their unpublished paper.

For $k = 3$ Theorem \ref{t2} establishes the relation between general geometric approach and well-known algebraic approach, which appeared in original works by Balinskii and Novikov \cite{balnov}. It implies that any such bracket corresponds to a Frobenius triple (see the exact definitions below) --- that is a commutative associative algebra with a pair of non-degenerate invariant forms. This implies, that in general the description of such brackets is a lost cause. Yet, in small dimension it is possible. Theorem \ref{t4} provides the classification of Darboux-Poisson brackets in dimension two. This classification contains the inhomogeneous brackets of many famous integrable systems. 

The last Theorem \ref{t3} provides an algebraic procedure to construct the conservation laws of such system in an implicit form. Similar to the case of dimension one we get an equation, bares similarities to both Balinskii-Novikov formula and Miura formula. This identity provides the recursion relation, which can be solved.

In the end of the section we provide the construction of Hunter-Saxton equation in case of arbitrary inhomogeneous Darboux-Poisson bracket with minor non-degeneracy conditions (formula \eqref{hs2}). We also describe two important examples, for which we write all the formulas explicitly. We compare them with the cases, known in the literature, and show, that they are, indeed, new. 

All the proofs are provided in sections \ref{proof1}, \ref{proof2}, \ref{proof3}, \ref{proof4}. The Appendix contains some important algebraic observations, which (to keep work self-sufficient) are provided with proofs.

\section{Main results of the paper}\label{main}

We start with general result on Darboux-Poisson brackets.

\begin{Theorem}\label{t1}
Consider homogeneous a non-degenerate Darboux-Poisson bracket of order $k$ in arbitrary coordinates $u^1, \dots, u^n$ with Hamiltonian operator
\begin{equation}\label{ty1}
\mathcal P^{\alpha \beta} = h^{\alpha \beta} D^k + \sum \limits_{j = 1}^k \mathcal P^{\alpha \beta}_{\mathrm{j}} D^{k - j}.    
\end{equation}
Recall, that $h^{\alpha \beta}$ is non-degenerate contrvariant tensor of rank two. We have
\begin{enumerate}
    \item The term $P_1^{q \beta}$ has the form 
    \begin{equation}\label{gs1}
         \mathcal P_1^{\alpha \beta} = - k h^{\alpha q}\Gamma^{\beta}_{qs} u^s_x,
    \end{equation}
    where $\Gamma^{\beta}_{qs}$ is flat symmetric connection.
    \item The contrvariant tensor $h^{\alpha \beta}$ is parallel along $\Gamma^{\beta}_{qs}$ and the flat coordinates of the connection are exactly the Darboux coordinates of the Darboux-Poisson bracket \eqref{ty1}.
    \item For $j \geq 1$ define matrices $\tau(j)^{\beta}_q$ using recursion formula
    \begin{equation}\label{gs2} 
    \tau(1)^{\beta}_q = - \Gamma^{\beta}_{qs} u^s_x, \quad \tau(j + 1)^{\beta}_q = \tau(1)^m_{q} \tau(j)^{\beta}_m + D\big(\tau(j)^{\beta}_q\big).
    \end{equation}
    Then the rest of the coefficients of $\mathcal P$ has the form
    \begin{equation}\label{coeff}
        \Ph^{\alpha \beta}_{\mathrm j} = \binom{j}{k} h^{\alpha q} \tau(j)^{\beta}_q, \quad j \geq 2.
    \end{equation}
\end{enumerate}
The inverse is also true: for a given collection of data 
\begin{enumerate}
    \item non-degenerate symmetric (skew-symmetric) contrvariant tensor $h^{\alpha \beta}$ of type $(2, 0)$;
    \item flat symmetric connection $\Gamma^\beta_{qs}$ with condition $\nabla h^{\alpha \beta} = 0$;
    \item odd (even) natural number $k$; 
\end{enumerate}
the formulas \eqref{gs1}, \eqref{coeff}, \eqref{gs2} and \eqref{ty1} define homogeneous non-degenerate Hamiltonian operator of Darboux-Poisson bracket of order $k$ in given coordinates.
\end{Theorem}

\begin{Remark}
\rm{
For even $k$ non-degenerate skew-symmetric tensor $h^{\alpha \beta}$ is actually the inverse of symplectic structure and $\Gamma^{\beta}_{rs}$ is flat symplectic connection. The corresponding pair defines the so called structure of flat Fedosov manifold \cite{ferg}, \cite{fed}, which is the main ingredient of Fedosov quantization procedure. To study the analogs of Harry Dym equations in this case is an interesting problem, but it is out of the scope of this paper. $\blacksquare$
}
\end{Remark}

In case of odd $k$ the Christoffel symbols $\Gamma^\beta_{rs}$ define Levi-Civita connection of corresponding flat metric $h$. Thus, all the ingredients of the Darboux-Poisson bracket are defined by single flat metric. So for Poisson bracket of odd order $k > 1$ we adopt the notion $\mathcal B_h$, where $h$ is corresponding flat metric. 

Note, that for $k = 1$, Theorem \ref{t1} yields the classical formulas from \cite{dn} by Dubrovin and Novikov
\begin{equation}\label{first}
    \mathcal A^{\alpha \beta}_g = g^{\alpha \beta} D - g^{\alpha q} \Gamma^{\beta}_{qs} u^s_x,    
    \end{equation}
with only change $h^{\alpha \beta} \to g^{\alpha \beta}$. For $k = 3$ the formula for components of the Hamiltonian operator of Darboux-Poisson bracket $\mathcal B_h$ is
\begin{equation}\label{third}
    \begin{aligned}
    \mathcal B^{\alpha \beta}_h & = h^{\alpha \beta} D^3 - 3 h^{\alpha q} \bar \Gamma^{\beta}_{q s} u^s_xD^2 + \\
    & + 3 \Bigg(h^{\alpha q}\Big(\bar \Gamma^p_{qs} \bar \Gamma^{\beta}_{pr} - \pd{\bar \Gamma^{\beta}_{qs}}{u^r}\Big)u^s_x u^r_x - h^{\alpha q}\bar \Gamma^{\beta}_{qs} u^s_{x^2} \Bigg) D + \\
    & + \Bigg( h^{\alpha q} \Big( 2 \bar \Gamma^a_{qs} \pd{\bar \Gamma^{\beta}_{ar}}{u^p} + \pd{\bar \Gamma^a_{qs}}{u^r}\bar \Gamma^{\beta}_{ap} - \bar \Gamma^a_{qs} \bar \Gamma^b_{ar} \bar \Gamma^{\beta}_{bp} - \frac{\partial^2 \bar \Gamma^{\beta}_{qs}}{\partial u^r \partial u^p}\Big) u^s_x u^r_x u^p_x + \\
    & + h^{\alpha q} \Big( 2 \bar \Gamma^a_{qs} \bar \Gamma_{ar}^{\beta} + \bar \Gamma^a_{qr} \bar \Gamma^{\beta}_{as} - 2 \pd{\bar \Gamma^{\beta}_{qr}}{u^s} - \pd{\bar \Gamma^{\beta}_{qs}}{u^r}\Big) u^s_x u^r_{x^2} - h^{\alpha q} \bar \Gamma^{\beta}_{qs} u^s_{x^3} \Bigg).
    \end{aligned}
\end{equation}
Here $\bar \Gamma^{\beta}_{qs}$ are Christoffel symbols of Levi-Civita connection of $h$ (we use notation with bar to distinguish formulas \eqref{first} and \eqref{third}). The inhomogeneous Poisson brackets $\mathcal P$
\begin{equation*}
\mathcal P = \mathcal A_g + \mathcal B_h   
\end{equation*}
is, thus, determined by a pair of flat metrics $g$ and $h$. Of course, these metrics are not arbitrary: the compatibility of $\mathcal A_g$ and $\mathcal B_g$ puts strong conditions on them. 

\begin{Theorem}\label{t2}
Let $\mathcal A_g$ be the Hamiltonian operator of homogeneous non-degenerate Poisson bracket of order one and $\mathcal B_h$ be the Hamiltonian operator of homogeneous non-degenerate Darboux-Poisson bracket of odd order $k$. The formula
\begin{equation}\label{sp}
\mathcal P^{\alpha \beta} = \mathcal A^{\alpha \beta}_g + \mathcal B^{\alpha \beta}_h   
\end{equation}
defines the Hamiltonian operator of inhomogeneous Poisson bracket of type $1 + k$ (equivalently, $\mathcal A_g$ and $\mathcal B_h$ are compatible), if and only if
\begin{itemize}
    \item[$k \geq 5$:] There exist coordinates $u^1, \dots, u^n$, in which both metrics $h$ and $g$ are in constant form. The Poisson bracket in these coordinates takes form 
    \begin{equation}\label{km1}
     \{\mathcal H_1, \mathcal H_2\} = \int \limits_{0}^{2 \pi} \Bigg( \vpd{\mathcal H_1}{u^{\alpha}}h^{\alpha \beta} D^k \Big( \vpd{\mathcal H_1}{u^{\beta}}\Big) + \vpd{\mathcal H_1}{u^{\alpha}}g^{\alpha \beta} D \Big( \vpd{\mathcal H_1}{u^{\beta}}\Big)\Bigg) \ddd x.
    \end{equation}
    \item[$k = 3$:] In coordinates $u^1, \dots, u^n$, where $h$ has constant entries $h^{\alpha \beta}$, the metric $g$ has the form
    $$
    g^{\alpha \beta} = 2(b^{\alpha \beta} + a^{\alpha \beta}_s u^s). 
    $$
    The constants $h^{\alpha \beta}, b^{\alpha \beta}, a^{\alpha \beta}_s$ are symmetric in upper indices and satisfy the conditions
     \begin{equation}\label{triple}
    \begin{aligned}
   a^{\alpha \beta}_q b^{q \gamma} = a^{\gamma \beta}_q b^{q \alpha}, \quad a^{\alpha \beta}_q a^{q \gamma}_s = a^{\gamma \beta}_q a^{q \alpha}_s, \quad a^{\alpha \beta}_q h^{q \gamma} = a^{\gamma \beta}_q h^{q \alpha}.
    \end{aligned}
    \end{equation}
    The Poisson bracket in these coordinates takes form
   \begin{equation}\label{km2}
   \begin{aligned}
   \{\mathcal H_1, \mathcal H_2\} = \int \limits_{0}^{2 \pi} \Bigg( \vpd{\mathcal H_1}{u^{\alpha}}h^{\alpha \beta} D^3 \Big( \vpd{\mathcal H_1}{u^{\beta}}\Big) + 2 \vpd{\mathcal H_1}{u^{\alpha}}(b^{\alpha \beta} + a^{\alpha \beta}_s u^s) D \Big( \vpd{\mathcal H_1}{u^{\beta}}\Big) + \vpd{\mathcal H_1}{u^{\alpha}} a^{\alpha \beta}_s u^s_x \vpd{\mathcal H_2}{u^{\beta}} \Bigg) \ddd x.
   \end{aligned}
    \end{equation}
\end{itemize}
\end{Theorem}

\begin{Corollary}\label{cor1}
Consider a pair of flat metrics $g^{\alpha \beta}$ and $h^{\alpha \beta}$ with Christoffel symbols of corresponding metrics being $\Gamma^{\alpha \beta}_s$ and $\bar \Gamma^{\alpha \beta}_s$ respectively and let $S^{\alpha \beta}_s = g^{\alpha q} \big( \Gamma^\beta_{qs} - \bar \Gamma^\beta_{qs}\big)$. Denote $\mathcal A_g$ and $\mathcal B_h$ to be the same Hamiltonian operators as in the statement of Theorem \ref{t2}. The formula \eqref{sp} defines the Hamiltonian operator of inhomogeneous Poisson bracket of type $1 + k$ if and only if 
\begin{itemize}
    \item[$k \geq 5$:] Tensor $S^{\alpha \beta}_s$ vanishes
    \item[$k = 3$:] The tensor $S^{\alpha \beta}_s$ is symmetric in upper indices and satisfies conditions
    \begin{equation}\label{cond}
    \begin{aligned}
        & \bar \nabla_q S^{\alpha \beta}_s = 0, \quad S^{\alpha \beta}_q S^{q \gamma}_s = S^{\gamma \beta}_q S^{q \alpha}_s, \quad S^{\alpha \beta}_q h^{q \gamma}_s = S^{\gamma \beta}_q h^{q \alpha}_s, \quad S^{\alpha \beta}_q b^{q \gamma}_s = S^{\gamma \beta}_q b^{q \alpha}_s.
    \end{aligned}
    \end{equation}
    In flat coordinates of $h^{\alpha \beta}$, the constants $a^{\alpha \beta}_s$ from formula \eqref{km2} are related to tensor $S^{\alpha \beta}_s$ as $a^{\alpha \beta}_s = - S^{\alpha \beta}_s$.
\end{itemize}
\end{Corollary}

The first statement of Theorem \ref{t2} shows, that in our line of investigation there are no non-trivial higher order analogs (for example, with fifth order derivative in $x$) of Harry Dym equation. The Hamiltonian operator for $k = 3$ in given dimension is defined by a finite collection of constants $a^{\alpha \beta}_s, b^{\alpha \beta}, h^{\alpha \beta}$. Let us explore the algebraic properties of these constants.

First, recall that flat coordinate system for $h$ is defined up to an affine transformation. In case of $k = 3$ the translation by vector $c = (c^1, \dots, c^n)$ changes $b^{\alpha \beta}$ into $\bar b^{\alpha \beta} = b^{\alpha \beta} - a^{\alpha \beta}_s c^s$. 

Let $\mathrm p$ be a coordinate origin of the coordinate system $u^1, \dots, u^n$, in which the inhomogeneou Poisson bracket is in the form \eqref{km2}. Restricting ourselves to the coordinate changes that preserve point $\mathrm p$, we observe that the constants $a^{\alpha \beta}_s$ behave as tensor of type $(2, 1)$, while $h^{\alpha \beta}, b^{\alpha \beta}$ behave as tensors of type $(2, 0)$.

Consider a commutative associative algebra $\mathfrak a, \star$ over $\mathbb R$ or $\mathbb C$. The symmetric bilinear form $h$ is invariant form if 
\begin{equation}\label{i1}
    h(\xi \star \eta, \zeta) = h(\xi, \eta \star \zeta), \quad \forall \xi, \eta, \zeta \in \mathfrak a.
\end{equation}
If an invariant form is non-degenerate, then it is called Frobenius form. A Frobenius triple $(\mathfrak a, b, h)$ is a commutative associative algebra $\mathfrak a$, equipped with two Frobenius forms $b$ and $h$. 

The formulas \eqref{triple} and \eqref{i1} imply that constants $a^{\alpha \beta}_s, b^{\alpha \beta}, h^{\alpha \beta}$ are just components of some Frobenius triple, written in a given basis of $T^* \mathrm M$. Moreover, the coordinate change that preserves coordinate origin $\mathrm p$ does not change the triple. So, in algebraic terms, the Theorem \ref{t2} establishes one-to-one correspondence between inhomogeneous Hamiltonian operators of the form \eqref{sp} around given point $\mathrm p \in \mathrm M$ and Frobenius triples on $T^*_{\mathrm p} \mathrm M^n$.
\begin{Remark}\label{rem3}
\rm{
In case of degenerate $h^{\alpha \beta}$ even if the Poisson bracket is in the form \eqref{km2} the constants $a^{\alpha \beta}_s$ do not necessarily define commutative associative algebra. The following is a Hamiltonian operator of inhomogeneous Poisson bracket (\cite{str}, the case of algebra N4)
$$
\mathcal P^{\alpha \beta} = \left(\begin{array}{cc}
     1 & 0  \\
     0 & 0
\end{array}\right) D^3 + \left(\begin{array}{cc}
     2 u^1 & u^2  \\
     u^2 & 0
\end{array}\right) D + \left(\begin{array}{cc}
     u^1_x & 0  \\
     u^2_x & 0
\end{array}\right).
$$
The corresponding algebra is non-commutative, but still associative (it is, in fact, an associative Novikov algebra). At the same time coordinate change $\bar u^2 = u^1, \bar u^1 = (u^2)^2$ transforms the Hamiltonian operator of the Poisson bracket into the form
$$
\bar{\mathcal P}^{\alpha \beta} = \left(\begin{array}{cc}
     0 & 0  \\
     0 & 1
\end{array}\right) D^3 + \left(\begin{array}{cc}
     0 & 2 \bar u^1  \\
     2 \bar u^1 & 2 \bar u^2
\end{array}\right) D + \left(\begin{array}{cc}
     0 & \bar u^1_x  \\
     \bar u^1_x & \bar u^2_x
\end{array}\right).
$$
Now the algebra, corresponding to $g^{\alpha \beta}$ is commutative and associative. This Hamiltonian operator can be obtained as a limit $h_1 \to 0$ in family $(31)$ (see Theorem \ref{t4}). $\blacksquare$
}
\end{Remark}

Theorem \ref{t2} implies, that for odd $k > 3$ the inhomogeneous Poisson bracket in the form \eqref{sp} are in one-to-one correspondence with a pair of symmetric forms $g, h$. The explicit description all such pairs for given dimension $n$ is cumbersome, but possible (see, for example, \cite{lan}).

In case of $k = 3$ the problem of classification of inhomogeneous Hamiltonian operators in the form \eqref{sp} contains the problem of classification commutative associative algebras with Frobenius forms. For arbitrary signature of the forms this problem seems to have no solution in general. 

Yet for any given dimension the classification is possible. In dimension two up to isomorphism there are five non-trivial associative commutative algebras, which we denote $\mathfrak a_1, \dots, \mathfrak a_5$. The trivial algebra we denote as $\mathfrak a_0$. In basis $\eta^1, \eta^2$ (for the consistency of the notation we assume that the coordinates on algebra have lower indices) the non-zero structure relations for each of them are
\begin{equation}\label{alg}
    \begin{aligned}
    \mathfrak a_1: & \quad \quad  \eta^2 \star \eta^2 = \eta^2, \\
    \mathfrak a_2: & \quad \quad  \eta^2 \star \eta^2 = \eta^1, \\
    \mathfrak a_3: & \quad \quad  \eta^2 \star \eta^2 = \eta^2, \quad \eta^1 \star \eta^2 = \eta^1, \\
    \mathfrak a_4: & \quad \quad  \eta^1 \star \eta^1 = \eta^1, \quad \eta^2 \star \eta^2 = \eta^2, \\
    \mathfrak a_5: & \quad \quad \eta^2 \star \eta^2 = \eta^2, \quad \eta^1 \star \eta^1 = - \eta^2, \quad \eta^1 \star \eta^2 = \eta^1.
    \end{aligned}
\end{equation}
The following Theorem describes the inhomogeneous Hamiltonian operators in dimension two.

\begin{Theorem}\label{t4}
Consider point $\mathrm p \in \mathrm M^2$. Every inhomogeneous Hamiltonian operator $\mathcal P$ of the form \eqref{sp} for odd $k$ can be brought around point $\mathrm p$ to exactly one normal form from the list below. The different values of parameters $h_1, h_2, b_1, b_2$ yield non-equivalent normal forms, the canonical coordinates $u^1, u^2$ are centered at $\mathrm p$:
\begin{itemize}
    \item[$(01)$:] $\left(\begin{array}{cc} h_1 & 0 \\ 0 & h_2\\ \end{array}\right) D^3 + \left(\begin{array}{cc} 1 & 0 \\ 0 & 1\\ \end{array}\right) D$
    \item[$(02)$:] $\left(\begin{array}{cc} h_1 & 0 \\ 0 & h_2\\ \end{array}\right) D^3 + \left(\begin{array}{cc} -1 & 0 \\ 0 & 1\\ \end{array}\right) D$
    \item[$(03)$:] $\left(\begin{array}{cc} h_1 & 0 \\ 0 & h_2\\ \end{array}\right) D^3 + \left(\begin{array}{cc} - 1 & 0 \\ 0 & - 1 \\ \end{array}\right) D$
    \item[$(04)$:] $\left(\begin{array}{cc} 0 & h_1 \\ h_1 & 1\\ \end{array}\right) D^3 + \left(\begin{array}{cc} 0 & 1 \\ 1 & 0\\ \end{array}\right) D$
    \item[$(05)$:] $\left(\begin{array}{cc} - h_2 & h_1 \\ h_1 & h_2\\ \end{array}\right) D^3 + \left(\begin{array}{cc} 0 & 1 \\ 1 & 0\\ \end{array}\right) D$
    \item[$(11)$:] $\left(\begin{array}{cc} h_1 & 0 \\ 0 & h_2\\ \end{array}\right) D^3 + 2 \left(\begin{array}{cc} 1 & 0 \\ 0 & b_2 + u^2 \\ \end{array}\right) D + \left(\begin{array}{cc} 0 & 0 \\ 0 & u^2_x\\ \end{array}\right)$
    \item[$(12)$:] $\left(\begin{array}{cc} h_1 & 0 \\ 0 & h_2\\ \end{array}\right) D^3 + 2 \left(\begin{array}{cc} - 1 & 0 \\ 0 & b_2 + u^2 \\ \end{array}\right) D + \left(\begin{array}{cc} 0 & 0 \\ 0 & u^2_x\\ \end{array}\right)$
    \item[$(21)$:] $\left(\begin{array}{cc} 0 & h_1 \\ h_1 & h_2\\ \end{array}\right) D^3 + 2 \left(\begin{array}{cc} 0 & 1 \\ 1 & u^1\\ \end{array}\right) D + \left(\begin{array}{cc} 0 & 0 \\ 0 & u^1_x\\ \end{array}\right)$
    \item[$(31):$] $\left(\begin{array}{cc} 0 & h_1 \\ h_1 & h_2\\ \end{array}\right) D^3 + 2 \left(\begin{array}{cc} 0 & 1 + u^1 \\ 1 + u^1 & b_2 + u^2\\ \end{array}\right) D + \left(\begin{array}{cc} 0 & u^1_x \\ u^1_x & u^2_x\\ \end{array}\right)$
    \item[$(41)$:]$\left(\begin{array}{cc} h_1 & 0 \\ 0 & h_2\\ \end{array}\right) D^3 + 2 \left(\begin{array}{cc} b_1 + u^1 & 0 \\ 0 & b_2 + u^2\\ \end{array}\right) D + \left(\begin{array}{cc} u^1_x & 0 \\ 0 & u^2_x\\ \end{array}\right)$
    \item[$(51)$:] $\left(\begin{array}{cc} - h_2 & h_1 \\ h_1 & h_2\\ \end{array}\right) D^3 + 2 \left(\begin{array}{cc} - b_2 - u^2 & b_1 + u^1 \\ b_1 + u^1 & b_2 + u^2\\ \end{array}\right) D + \left(\begin{array}{cc} - u^2_x & u^1_x \\ u^1_x & u^2_x\\ \end{array}\right)$
\end{itemize}
\end{Theorem}

The families of normal forms in Theorem \ref{t4} contain brackets of many famous two-component integrable systems: Hirota-Satsuma equation (family $41$), Antonowitcz-Fordy dispersive water equation (family $31$), Ito system (family $41$) just to name the few. 

Now let us go back to the arbitrary dimension $n$. The Hamiltonian operators $\mathcal A_g$ and $\mathcal B_h$ are the same as in the statement of the Theorem \ref{t2}. We call the sequence of functionals $\mathcal H_i, i \geq 1$ the multicomponent Harry Dym hierarchy if it satisfies the conditions
\begin{equation}\label{hd2}
\begin{aligned}
 u^\alpha_{t_0} & = \mathcal A^{\alpha q}_g \vpd{\mathcal H_1}{u^q} = 0, \\
 u^\alpha_{t_i} & = \mathcal A^{\alpha q}_g \vpd{\mathcal H_{i + 1}}{u^q} = \mathcal B^{\alpha q}_h \vpd{\mathcal H_i}{u^q}, \quad i \geq 2.
\end{aligned}
\end{equation}
The equation $u^\alpha_{t_1}$ we call multicomponent Harry Dym equation.

We say, that matrix $L$ is a good square root of $R$ if there exists a polynomial $p(t)$ (may be with complex coefficients), such that $p(R) = L$ and $L^2 = M$. If $R$ is non-degenerate, then at least one good root always exists (Lemma \ref{wwr}).

The next theorem shows, that in coordinates, described in Theorem \ref{t2}, the densities for the multicomponent Harry Dym hierarchy \ref{hd2} can be found by purely algebraic means, without any integration.

\begin{Theorem}\label{t3}
Let $a^{\alpha \beta}_s, b^{\alpha \beta}, h^{\alpha \beta}$ be a Frobenius triple, written in a given basis. Define $r^{\beta}_p = \frac{1}{2}  h^{\beta q} b_{qp}$ and $c^\beta_{ps} = \frac{1}{2} b_{pq} a^{q \beta}_s$, where $b_{pq}$ are components of the inverse of $b^{\alpha \beta}$. Denote $l^\beta_s$ to be the components of a good square root of matrix $r^\beta_p$.

Consider a system of algebraic equations
\begin{equation}\label{rq2}
    u^\beta = \mathrm v^{\beta} + \frac{1}{2} c^\beta_{ps} \mathrm v^p \mathrm v^s + l^\beta_s \mathrm v^s_x, \quad \beta = 1, \dots, n.
\end{equation}
The solution to this system is a collection of series $\mathrm v^1, \dots, \mathrm v^n$ in the form $\mathrm v^\beta = \mathcal v^\beta_0 + \mathcal v^\beta_1 + \mathcal v^\beta_2 + \dots$. Here $\mathcal v^\beta_i$ are either zero or a homogeneous differential polynomial of degree $i - 1$. Take
\begin{equation}\label{rq3}
    \mathcal H^\alpha_i = \int \limits^{2 \pi}_{0} v^\alpha_{2i - 1} \ddd x, \quad i = 1, \dots.
\end{equation}
Then the following holds.
\begin{enumerate}
    \item For each $\alpha = 1, \dots, n$ the series $\mathcal H^\alpha_i$ defines a multicomponent Harry Dym hierarchy
    \begin{equation}\label{hd3}
    \begin{aligned}
    u^\beta_{t_0^\alpha} & = \mathcal A^{\beta q}_g \vpd{\mathcal H^\alpha_1}{u^q} = 0, \\
    u^\beta_{t_i^\alpha} & = \mathcal A^{\beta q}_g \vpd{\mathcal H^\alpha_{i + 1}}{u^q} = \mathcal B_h^{\beta q} \vpd{\mathcal H^{\alpha}_{i + 1}}{u^q}, \quad i \geq 2.
    \end{aligned}
    \end{equation}
    for the Hamiltonian operators $\mathcal A^{\beta q}_g = 2 (b^{\beta q} + a^{\beta q}_s u^s) D + a^{\beta q}_s u^s_x$ and $\mathcal B_h^{\beta q} = h^{\beta q} D^3$
    \item For arbitrary $\alpha, \beta$ and $i, j$ the flows $u_{t^\alpha_i}$ and $u_{t^\beta_j}$ commute. The functional $\mathcal H^\alpha_i$ is a conservation law of $u_{t^\beta_j}$
    \item Consider the sequence of functionals $\mathcal H_i, i \geq 1$, which define the multicompomemt Harry Dym hierarchy with respect to the the Hamiltonian operators $\mathcal A^{\beta q}_g = 2 (b^{\beta q} + a^{\beta q}_s u^s) D + a^{\beta q}_s u^s_x$ and $\mathcal B_h^{\beta q} = h^{\beta q} D^3$. Assume in addition, that each $\mathcal H_i$ possesses the following property: it is either zero or a homogeneous differential polynomial of degree $2i - 1$. Then there exists a collection of constants $c_0, c_1, \dots, c_n$, such that
    \begin{equation*}
        \begin{aligned}
        \mathcal H_1 & = c_0 + c_1 \mathcal H^1_1 + \dots + c_n \mathcal H^n_1, \\
        \mathcal H_i &= c_1 \mathcal H^1_i + \dots + c_n \mathcal H^n_i, \quad i > 1.
        \end{aligned}
    \end{equation*}
    In particular, the linear space of the multicomponent Harry Dym hierarchies with the property given above is finite-dimensional and sequences $\mathcal H^\alpha_i, i \geq 1, \alpha = 1, \dots, n$ together with constant functional form a basis in this space.
    \item If in addition $a^{\alpha \beta}_s$ has unity, then one can replace equation \eqref{rq2} with simpler Miura type equation
    \begin{equation}\label{rq2_2}
    u^\beta = \frac{1}{2} c^\beta_{ps} \mathrm v^p \mathrm v^s + l^\beta_s \mathrm v^s_x, \quad \beta = 1, \dots, n.
    \end{equation}
    and Hamiltonian operator of order one in formula \eqref{hd3} with $\mathcal A^{\beta q}_g = 2 a^{\beta q}_s u^s D + a^{\beta q}_s u^s_x$.
\end{enumerate}
\end{Theorem}

Assume now that in Frobenius triple $a^{\alpha \beta}_s, b^{\alpha \beta}, h^{\alpha \beta}$ the commutative associative algebra has unity $e = (e_1, \dots, e_n)$. For Hamiltonian operators $\mathcal A^{\beta q}_g = 2 a^{\beta q}_s u^s D + a^{\beta q}_s u^s_x$ and $\mathcal B_h^{\beta q} = h^{\beta q} D^3$ (these are the Hamiltonian operators from the fourth statement of the Theorem \ref{t3}) we write negative hierarchy for $i = 0, - 1, - 2$ we get:
\begin{equation}\label{neg2}
\begin{aligned}
    u^\beta_{t_0} & = h^{\beta q} D^3 \Big(\vpd{\bar{\mathcal H}_0}{u^q}\Big) = 0, \\
    u^\beta_{t_{-1}} & = 2 a^{\beta q}_s u^s D \Big( \vpd{\bar{\mathcal H}_0}{u^q}\Big) + a^{\beta q}_s u^s_x \vpd{\bar{\mathcal H}_0}{u^q} = h^{\beta q} D^3 \Big(\vpd{\bar{\mathcal H}_{-1}}{u^q}\Big), \\
    u^\beta_{t_{-2}} & =  2 a^{\beta q}_s u^s D \Big( \vpd{\bar{\mathcal H}_{-1}}{u^q}\Big) + a^{\beta q}_s u^s_x \vpd{\bar{\mathcal H}_{-1}}{u^q}
\end{aligned}
\end{equation}
Consider functional $\bar{\mathcal H}_0 = \int^{2 \pi}_{0} e_\alpha u^\alpha \ddd x$. The equation for $i = -1$ takes form
$$
u^\alpha_{t_{-1}} = u^\alpha_x = h^{\alpha q} D^3 \Big(\vpd{\bar{\mathcal H}_{-1}}{u^q}\Big). 
$$
We introduce new variables $v^\alpha$ as
$$
v^p_{xxx} = \mathcal B^{pq} \vpd{\bar{\mathcal H}_{-1}}{u^q}.
$$
In given coordinates this yields $h^{qp} \vpd{\bar{\mathcal H}_{-1}}{u^q} = v^p$. We get, that $u^\alpha = v^\alpha_{xx}$. The equation \eqref{neg2} for $i = - 2$ takes form 
$$
v^\alpha_{xxt_{-2}} = 2 a^{\alpha q}_s v^s_{xx} h_{qp} v^p_x + a^{\alpha q}_s v^s_{xxx} h_{qp} v^p. 
$$
Denote $\bar c^\alpha_{ps} = h_{pq} a^{\alpha q}_s$. Renaming $t_{-2}$ as $t$, we integrate both sides of the equation, using property $\bar c^\alpha_{ps} = \bar c^\alpha_{sp}$ (Lemma \ref{lemm3}). Taking the integration constants we arrive to the system
\begin{equation}\label{hs2}
    \begin{aligned}
     v^\alpha_{xt} = \frac{1}{2} \bar c^\alpha_{ps} \big( v^p_{xx} v^s + v^p_x v^s_x + v^p v^s_{xx}\big).
    \end{aligned}
\end{equation}
We call this system the multicomponent Hunter-Saxton equation.

By construction we see, that these equations are defined for arbitrary commutative associative algebra. The existence of Frobenius form is necessary and sufficient for this equation to be bihamiltonian with respect to $\mathcal A_g$ and $\mathcal B_g$. The conservation laws for this equation are obtained from Theorem \ref{t3}.

In the end of the paper we consider two examples. 


\subsection{Example I: Multicomponent Harry Dym and Hunter-Saxton equations, associated with algebra $\mathfrak t_n$} Following \cite{str} we define commutative associative algebra with unity $\mathfrak t_n$: in given basis $\eta^1, \dots, \eta^n$ the multiplication is defined as 
\begin{equation}\label{tn}
\eta^i \star \eta^j = \delta^{i + j - 1}_k \eta^k.     
\end{equation}
We have that $\eta^1$ is unity and for $k \geq 2$
$$
\eta^k = \underbrace{\eta^2 \star \dots \star \eta^2}_{\text{$k - 1$ times}}.
$$
All the calculations in Theorem \ref{t3} for algebra $\mathfrak t_n$ can be done explicitly. 

Assume that we have different basis $\bar \eta^i$, such that in dual basis $\eta_i = \frac{1}{2} b_{iq} \bar \eta^q$ the structure constants $c^\beta_{ps}$ are in the form \eqref{tn}. Matrix $U$ is defined as $U \eta = \eta \star u$, where $u = \eta_1 u^1 + \dots + \eta_n u^n$. In given coordinates 
\begin{equation}\label{umat}
U = \left( \begin{array}{cccccc}
     u^1 & 0 & 0 & \dots & 0 & 0  \\
      u^2 & u^1 & 0 & \dots & 0 & 0 \\
      u^3 & u^2 & u^1 & \dots & 0 & 0 \\
      & & & \ddots & & \\
      u^{n - 1} & u^{n - 2} & u^{n - 3} & \dots & u^1 & 0 \\
      u^n & u^{n - 1} & u^{n - 2} & \dots & u^2 & u^1
\end{array}\right).
\end{equation}
The second matrix we need for calculation is
\begin{equation}\label{vmat}
\mathtt V = \left( \begin{array}{cccccc}
     \mathrm v^1 & 0 & 0 & \dots & 0 & 0  \\
      \mathrm v^2 & \mathrm v^1 & 0 & \dots & 0 & 0 \\
      \mathrm v^3 & \mathrm v^2 & \mathrm v^1 & \dots & 0 & 0 \\
      & & & \ddots & & \\
      \mathrm v^{n - 1} & \mathrm v^{n - 2} & \mathrm v^{n - 3} & \dots & \mathrm v^1 & 0 \\
      \mathrm v^n & \mathrm v^{n - 1} & \mathrm v^{n - 2} & \dots & \mathrm v^2 & \mathrm v^1
\end{array}\right).
\end{equation}
The components of this matrix are series $\mathrm v^\alpha = \mathcal v^1 + \mathcal v^2 + \dots$ (see Theorem \ref{t3}). The matrix $\mathtt V$ naturally decomposes into a series $\mathtt V = V_1 + V_2 + \dots$, where entries of $V_i$ are either zero or $\mathcal v^\alpha_i$.

In terms of $U$ and $\mathtt V$ the equation \eqref{rq2_2} can be written as
\begin{equation}\label{alg1}
U = \frac{1}{2} \mathtt V^2 + \mathtt V_x,    
\end{equation}
where the terms of matrix $\mathtt V_x$ are the derivative in $x$ of the corresponding terms of $\mathtt V$. This system looks exactly like the system \eqref{ric} in one-dimensional case with scalars being replaced by matrices. The equation \eqref{alg1} produces a system of algebraic equations on $V_i$
\begin{equation}\label{alg2}
\begin{aligned}
U & = \frac{1}{2} V_1^2, \\
V_{i + 1} & = - V_1^{-1} \Bigg(\,\frac{1}{2} \sum \limits_{j = 2}^{i} V_j V_{i + 2 - j} + (V_{i})_x \Bigg), \quad i \geq 1. 
\end{aligned}
\end{equation}
To solve the first equation in \eqref{alg2} recall, that
\begin{equation}\label{taylor}
(1 + t)^{\alpha} = \sum \limits_{i = 0}^{2 \pi} \binom{\alpha}{i} t^i, \quad \binom{\alpha}{i} = \frac{\alpha (\alpha - 1) \dots (\alpha - i + 1)}{i!}.    
\end{equation}
Using this notation we get
\begin{equation}\label{alg3}
V_1 = \sqrt{2 u^1} \Big(\operatorname{Id} + \frac{1}{u_1} (U - u_1 \operatorname{Id}) \Big)^{\frac{1}{2}} = \sqrt{2 u^1} \Bigg(\sum \limits_{i = 0}^{n - 1} \binom{\frac{1}{2}}{i}\Big( \frac{1}{u^1} U - \operatorname{Id}\Big)^i \Bigg).    
\end{equation}
Here we used that $U - u_1 \operatorname{Id}$ is nilpotent. The system \ref{alg2} takes form
\begin{equation}\label{solv}
\begin{aligned}
V_1 & = \sqrt{2 u^1} \Bigg(\sum \limits_{i = 0}^{n - 1} \binom{\frac{1}{2}}{i}\Big( \frac{1}{u^1} U - \operatorname{Id}\Big)^i \Bigg),  \\
V_{i + 1} & = - \frac{1}{\sqrt{2u^1}} \Bigg(\sum \limits_{i = 0}^{n - 1} \binom{- \frac{1}{2}}{i}\Big( \frac{1}{u^1} U - \operatorname{Id}\Big)^i \Bigg) \Bigg(\,\frac{1}{2} \sum \limits_{j = 2}^{i} V_j V_{i + 2 - j} + (V_{i})_x \Bigg), \quad i \geq 1. 
\end{aligned}
\end{equation}
In case of $n = 2$ the formula \eqref{alg3} produces 
$$
\mathcal H^1_1 = \int^{2 \pi}_{0} \sqrt{2 u^1} \ddd x, \quad \mathcal H^2_1 = \int^{2 \pi}_{0} \frac{u^2}{\sqrt{2 u^1}} \ddd x.
$$
These are the Hamiltonians of the coupled Harry Dym equation, obtained by Antonowicz and Fordy in \cite{fordy2}. For $n = 3$ the formula \eqref{alg3} produces
$$
\mathcal H^1_1 = \int^{2 \pi}_{0} \sqrt{2 u^1} \ddd x, \quad \mathcal H^2_1 = \int^{2 \pi}_{0} \frac{u^2}{\sqrt{2 u^1}} \ddd x, \quad \mathcal H^3_1 = \int^{2 \pi}_{0} \Bigg( \frac{u^3}{(2 u^1)^{\frac{1}{2}}} - \frac{(u^2)^2}{(2 u^1)^{\frac{3}{2}}}\Bigg) \ddd x.
$$
Note that $\mathcal A_g$ in the fourth statement of Theorem \ref{t3} does not depend on $b^{\alpha \beta}$. Thus, we may assume that $2 b^{\alpha \beta} = h^{\alpha \beta}$. In particular, the constants $\bar c^\alpha_{ps}$, that appear in construction of multicomponent Hunter-Saxton equation, coincide with $c^\alpha_{ps}$. We introduce matrix
\begin{equation}\label{vmat2}
V = \left( \begin{array}{cccccc}
     v^1 & 0 & 0 & \dots & 0 & 0  \\
      v^2 & v^1 & 0 & \dots & 0 & 0 \\
      v^3 & v^2 & u^1 & \dots & 0 & 0 \\
      & & & \ddots & & \\
      v^{n - 1} & v^{n - 2} & v^{n - 3} & \dots & v^1 & 0 \\
      v^n & v^{n - 1} & v^{n - 2} & \dots & v^2 & v^1
\end{array}\right).
\end{equation}
Basically matrix \eqref{vmat2} is matrix \eqref{umat} with $u^\alpha$ replaced by $v^\alpha$. By construction of coordinates $v^i$ we have $V_{xx} = U$. The multicomponent Hunter-Saxton equation can be written as
\begin{equation}\label{hsmat}
    V_{xt} = V_{xx} V + \frac{1}{2} V^2_x.
\end{equation}
The system \eqref{hsmat} seems to be new. For $n = 2$ we get
$$
\begin{aligned}
v^1_{xt} & = v^1_{xx} v^1 + \frac{1}{2} (v^1_x)^2, \\
v^2_{xt} & = v^1_{xx} v^2 + v^1_x v^2_x + v^1 v^2_{xx}.
\end{aligned}
$$
One notices the differences between this multicomponent version of Hunter-Saxton equation and the so-called 2-HS equation. The latter was obtained as a short-wave limit of 2-CH equation (see \cite{mult} for formulas and further references on the subject and its applications). For $n = 3$ the multicomponent Hunter-Saxton equation in coordinates $v^1, v^2, v^3$ takes form
$$
\begin{aligned}
\begin{aligned}
v^1_{xt} & = v^1_{xx} v^1 + \frac{1}{2} (v^1_x)^2, \\
v^2_{xt} & = v^1_{xx} v^2  + v^1_x v^2_x + v^2 v^1_{xx}, \\
v^3_{xt} & = v^2_{xx} v^2 + \frac{1}{2} (v^2_x)^2 + v^1 v^3_{xx} + v^1_{x} v^3_x + v^1_{xx} v^3.
\end{aligned}
\end{aligned}
$$


\subsection{Example II: The Hamiltonian operators, multicomonent Harry Dym and Hunter-Saxton equations, associated with certain commutative associative algebra in dimension $4$}

Consider a commutative associative algebra with unity and Frobenius form. Proposition \ref{pp3} from Appendix implies, that every such an algebra is decomposed into the direct sum of irreducible terms. For dimensions two and three the only terms, that appear in such a decomposition are $\mathfrak t_1, \mathfrak t_2, \mathfrak t_3$ and $\mathfrak t_1^{\mathbb C}$ ($\mathfrak t_1^{\mathbb C}$ stands for the real form of complex one-dimensional algebra $\mathfrak t_1$).

The first commutative associative algebra that differs from the aforementioned series of examples, appears in dimension $4$. In given basis $\eta^1, \dots, \eta^4$ define the non-zero structure relations as
$$
\begin{aligned}
\eta^1 \star \eta^i & = \eta^i, \quad i = 1, \dots, 4, \\
\eta^2 \star \eta^2 & = \eta^4, \quad \eta^3 \star \eta^3 = \eta^4.
\end{aligned}
$$
To simplify the calculations, we take $h^1 = 1, h_i = 0, i = 2, 3, 4$ and define the Frobenius forms as $2 b^{\alpha \beta} = h^{\alpha \beta} = a^{\alpha \beta}_s h^s$. 

The Hamiltonian operator of the inhomogeneous Poisson bracket of type $1 + 3$, associated with this Frobenius triple in coordinates, defined by Theorem \ref{t2}, is
\begin{equation*}
    \mathcal P^{\alpha \beta} = \left(\begin{array}{cccc}
              0 & 0 & 0 & 1  \\
              0 & 1 & 0 & 0 \\
              0 & 0 & 1 & 0 \\
              1 & 0 & 0 & 0 \\
        \end{array}
        \right) D^3 +  2\left(\begin{array}{cccc}
             u^1 & u^2 & u^3 & u^4  \\
              u^2 & u^4 & 0 & 0 \\
              u^3 & 0 & u^4 & 0 \\
              u^4 & 0 & 0 & 0 \\
        \end{array}
        \right) D + \left(\begin{array}{cccc}
             u_x^1 & u_x^2 & u_x^3 & u_x^4  \\
              u_x^2 & u_x^4 & 0 & 0 \\
              u_x^3 & 0 & u_x^4 & 0 \\
              u_x^4 & 0 & 0 & 0 \\
        \end{array}
        \right).
\end{equation*}
Lowering the indices $\eta_i = h_{iq} \eta^i$ we obtain a commutative associative structure on tangent space
$$
\begin{aligned}
\eta_4 \star \eta_i & = \eta_i, \quad i = 1, \dots, 4, \\
\eta_2 \star \eta_2 & = \eta_1, \quad \eta_3 \star \eta_3 = \eta_1.
\end{aligned}
$$
We apply Theorem \ref{t3} to find the Casimirs of this bracket. To do that we take the equation \eqref{rq2_2} with $l^\beta_p = \delta^\beta_p$ takes form
\begin{equation}\label{4dim}
\begin{aligned}
    u^1 & = \mathrm v^1 \mathrm v^4 + \frac{1}{2} (\mathrm v^2)^2 + \frac{1}{2} (\mathrm v^3)^2 + \mathrm v^1_x, \\
    u^2 & = \mathrm v^4 \mathrm v^2 + \mathrm v^2_x, \\
    u^3 & = \mathrm v^4 \mathrm v^3 + \mathrm v^3_x, \\
    u^4 & = \frac{1}{2} (\mathrm v^4)^2 + \mathrm v^4_x.
\end{aligned}
\end{equation}
We get the following collection of series of functionals $\mathcal H^\alpha_1$:
$$
\mathcal H^4_1 = \int \limits^{2 \pi}_{0} \sqrt{2 u^4} \ddd x, \quad \mathcal H^3_1 = \int \limits^{2 \pi}_{0} \frac{u^3}{\sqrt{2u^4}} \ddd x, \quad \mathcal H^2_1 = \int \limits^{2 \pi}_{0} \frac{u^2}{\sqrt{2u^4}} \ddd x, \quad \mathcal H^1_1 = \int \limits^{2 \pi}_{0} \Bigg( \frac{u^1}{(2 u^4)^{\frac{1}{2}}} - \frac{1}{2} \frac{(u^2)^2 + (u^3)^2}{(2 u^4)^{\frac{3}{2}}} \Bigg) \ddd x.
$$
For $\alpha = 1, i = 1$ the bihamiltonian flow \eqref{hd3} of Theorem \ref{t3} is written as
\begin{equation}\label{hd4}
\begin{aligned}
    u^1_t & = \Bigg( \frac{3}{2} \frac{(u^2)^2 + (u^3)^2}{(2u^4)^{\frac{5}{2}}} - \frac{u^1}{(2u^4)^{\frac{3}{2}}}\Bigg)_{xxx},  \\
     u^2_t & = - \Bigg( \frac{u^2}{(2 u^4)^{\frac{3}{2}}}\Bigg)_{xxx},  \\
     u^3_t & = - \Bigg( \frac{u^3}{(2 u^4)^{\frac{3}{2}}}\Bigg)_{xxx},  \\
     u^4_t & = \Bigg( \frac{1}{\sqrt{2 u^4}}\Bigg)_{xxx}.
\end{aligned}
\end{equation}
Here we renamed $t^1_1$ as $t$. The equation \eqref{hd4} is the four-component Harry Dym equation. The corresponding 4-component Hunter-Saxton equations is written as
\begin{equation}
\begin{aligned}
v^1_{xt} & = v^1_{xx} v^4 + v^1_x v^4_x + v^1 v^4_{xx} + v^2_{xx} v^2 + \frac{1}{2} (v^2_x)^2 + v^3_{xx} v^3 + \frac{1}{2} (v^3_x)^2, \\
v^2_{xt} & = v^4_{xx} v^2 + v^4_x v^2_x + v^4 v^2_{xx}, \\
v^3_{xt} & = v^4_{xx} v^3 + v^4_x v^3_x + v^4 v^3_{xx}, \\
v^4_{xt} & = v^4_{xx} v^4 + \frac{1}{2} (v^4)^2
\end{aligned}    
\end{equation}
This system looks like two entangled two-component Harry Dym systems from the previous example.


\section{Proof of Theorem \ref{t1}}\label{proof1}

Assume that $\mathcal P^{\alpha \beta}$ is written in Darboux coordinates $u^1, \dots, u^n$. Denote $\Gamma^{\beta}_{rs}$ to be the Christoffel symbols of the flat symmetric connection, for which these coordinates are flat.

Consider another coordinate system $\bar u^1, \dots, \bar u^n$ and denote the Christoffel symbols of the same connection in these coordinates as $\bar \Gamma^{\beta}_{rs}$. We have
\begin{equation*}
\begin{aligned}
    0 = D(\delta^{\beta}_{r}) & = D\Big( \pd{\bar u^{\beta}}{u^{p}} \pd{u^{p}}{\bar u^{r}}\Big) = D\Big( \pd{\bar u^{\beta}}{u^{p}} \Big) \pd{u^{p}}{\bar u^{r}} +  \pd{\bar u^{\beta}}{u^{p}} D\Big(\pd{u^{p}}{\bar u^{r}}\Big) = \\
    & = D\Big( \pd{\bar u^{\beta}}{u^{p}} \Big) \pd{u^{p}}{\bar u^{r}} +  \pd{\bar u^{\beta}}{u^{p}} \frac{\partial^2 u^{p}}{\partial \bar u^{r} \partial \bar u^s} \bar u^s_{x} = D\Big( \pd{\bar u^{\beta}}{u^{p}} \Big) \pd{u^{p}}{\bar u^{r}} + \bar \Gamma^{\beta}_{r s} \bar u^s_{x}.
\end{aligned}
\end{equation*}
Here we have used the standard transformation rule for the Christoffel symbols. After multiplying by Jacobi matrix $\pd{\bar u}{u}$ both sides and rearranging the terms, we get formula
\begin{equation}\label{f3}
    D\Big(\pd{\bar u^{\beta}}{u^{p}} \Big) = - \pd{\bar u^q}{u^{p}} \bar \Gamma^{\beta}_{qs}\bar u^s_{x} = \pd{\bar u^q}{u^{p}} \tau^\beta_q (1). 
\end{equation}
Let us show that the more general formula holds
\begin{equation}\label{f4}
    D^j \Big(\pd{\bar u^{\beta}}{u^{p}} \Big) = \pd{\bar u^q}{u^{p}} \tau^\beta_q (j), \quad j \geq 1.
\end{equation}
We prove it by induction. The basis is provided by formula \eqref{f3}. The induction step is 
\begin{equation*}
    \begin{aligned}
    D^{j + 1} \Big(\pd{\bar u^{\beta}}{u^{p}} \Big) & = D \Big( \pd{\bar u^q}{u^{p}} \tau_{q}^{\beta}(j)\Big) = D \Big( \pd{\bar u^{m}}{u^{q}}\Big) \tau^\beta_m (j) + \pd{\bar u^{q}}{u^{p}} D\big( \tau^\beta_q (j)\big) = \\
    & = \pd{\bar u^q}{u^p} \tau^m_q(1) \tau^\beta_m(j) + \pd{\bar u^{q}}{u^{p}} D\big( \tau^\beta_q (j)\big) = \pd{\bar u^q}{u^p} \tau^\beta_q (j + 1).
    \end{aligned}
\end{equation*}
The formula \eqref{f4} is proved.

Now consider a pair of functionals $\mathcal H_1 = \int_{0}^{2 \pi} \mathcal h_1 \ddd x$ and $\mathcal H_2 = \int_{0}^{2 \pi} \mathcal h_2 \ddd x$. By direct computation we get
\begin{equation}\label{f5}
    \begin{aligned}
    \{\mathcal H_1, \mathcal H_2\} & = \int \limits^{2 \pi}_{0} h^{pq}\vpd{\mathcal h_1}{u^p} D^k \Big(\vpd{\mathcal h_2}{u^q}\Big) \ddd x = \int \limits^{2 \pi}_{0} h^{pq} \pd{\bar u^{\alpha}}{u^p} \vpd{\mathcal h_1}{\bar u^{\alpha}} D^k \Big(\vpd{\mathcal h_2}{\bar u^{\beta}} \pd{\bar u^{\beta}}{u^q}\Big) \ddd x = \\
    & = \int \limits^{2 \pi}_{0} h^{pq} \pd{\bar u^{\alpha}}{u^p} \vpd{\mathcal h_1}{\bar u^{\alpha}} \Bigg(\sum \limits_{j = 0}^k \binom{j}{k} D^j \Big( \pd{\bar u^{\beta}}{u^q}\Big)D^{k - j} \Big(\vpd{\mathcal h_2}{\bar u^{\beta}}\Big)\Bigg) \ddd x = \\
    & = \int \limits^{2 \pi}_{0} \Bigg(\bar h^{\alpha \beta} \vpd{\mathcal h_1}{\bar u^{\alpha}}\vpd{\mathcal h_2}{\bar u^{\beta}} + \sum \limits_{j = 1}^k \binom{j}{k} h^{\alpha q} \tau^{\beta}_q(j) \vpd{\mathcal h_1}{\bar u^{\alpha}} D^{k - j} \Big(\vpd{\mathcal h_2}{\bar u^{\beta}}\Big)\Bigg)\ddd x = \\
    & = \int \limits^{2 \pi}_{0} \vpd{\mathcal H_1}{u^{\alpha}} h^{\alpha \beta} D^k \Big(\vpd{\mathcal H_2}{u^{\beta}} \Big) \ddd x + \int \limits^{\infty}_{0} \vpd{\mathcal H_1}{u^{\alpha}} \sum \limits_{j = 1}^k \mathcal P^{\alpha \beta}_{\mathrm{j}} D^{k - j} \Big(\vpd{\mathcal H_2}{u^{\beta}}\Big) \ddd x
    \end{aligned}
\end{equation}
Here we applied formula \eqref{f4}. 

Recall, that $\binom{1}{k} = k$. This implies, that $\mathcal P_1^{\alpha \beta} = - k h^{\alpha q} \Gamma^{\beta}_{qs} u^s_x$. Thus, the first statement of the Theorem \ref{t1} holds. 

The coordinates $u^1, \dots, u^n$ are the Darboux coordinates for Poisson bracket and flat coordinates for $\Gamma^\beta_{qs}$. By construction
$$
0 = \pd{h^{\alpha \beta}}{u^s} = \nabla_s h^{\alpha \beta},
$$
thus, the second statement of Theorem \ref{t1} is true. The third statement of the Theorem follows directly from \eqref{f4}.

Finally, assume that for a given odd $k$ (for even the proof is similar) we have a non-degenerate symmetric tensor $h^{\alpha \beta}$ of type $(2, 0)$ and flat symmetric connection $\Gamma^\beta_{qs}$ with condition $\nabla h^{\alpha \beta} = 0$. Taking $u^1, \dots, u^n$ to be the flat coordinates of the connection, we arrive to the situation we have started the proof with. Thus, the theorem is proved.


\section{Proof of Theorem \ref{t2} and Corollary \ref{cor1}} \label{proof2}

First, let us show, that the conditions of Theorem \ref{t2} are necessary. 

Fix Darboux coordinates $u^1, \dots, u^n$ of $\mathcal B_h$. The Christoffel symbols of the Levi-Civita connection of metric $g$ in these coordinates are denoted as $\Gamma^\beta_{qs}$. The contrvariant Christoffel symbols are defined as (our definition differs by the sign from the one in \cite{dub1}) $\Gamma^{\alpha \beta}_s = g^{\alpha q} \Gamma^\beta_{qs}$. The condition $\nabla_q g^{\alpha \beta} = 0$ yields
\begin{equation}\label{g1}
    \pd{g^{\alpha \beta}}{u^q} + \Gamma^{\alpha \beta}_q + \Gamma^{\beta \alpha}_q = 0.
\end{equation}
The symmetry of Levi-Civita connection yields the second condition
\begin{equation}\label{g2}
   \Gamma^{\alpha \beta}_q g^{q \gamma} = \Gamma^{\gamma \beta}_q g^{q \alpha}.
\end{equation}

Consider the functionals
$$
\mathcal H_1 = \int \limits^{2 \pi}_{0} \mathcal a_\alpha u^\alpha \ddd x, \quad \mathcal H_2 = \int \limits^{2 \pi}_{0} \mathcal b_\beta u^\beta \ddd x, \quad \mathcal H_3 = \int \limits^{2 \pi}_{0} \mathcal c_\gamma u^\gamma \ddd x
$$
where $\mathcal a_{\alpha}, \mathcal b_{\beta}, \mathcal c_{\gamma}$ for $\alpha, \beta, \gamma = 1, \dots, n$ are arbitrary $3n$ functions of variable $x$. We introduce notation $\mathcal a^{(i)}_{\alpha} = D^i (\mathcal a_{\alpha})$ and same for $\mathcal b_{\beta}, \mathcal c_{\gamma}$. In these notations the compatibility condition for corresponding brackets is written as follows: for arbitrary functionals $\mathcal H_1, \mathcal H_2, \mathcal H_3$ one has
\begin{equation}\label{g3}
    \begin{aligned}
    & 0 = \{\{\mathcal H_1, \mathcal H_2\}_g, \mathcal H_3\}_h + \{\{\mathcal H_2, \mathcal H_3\}_g, \mathcal H_1\}_h + \{\{\mathcal H_3, \mathcal H_1\}_g, \mathcal H_2\}_h + \{\{\mathcal H_1, \mathcal H_2\}_h, \mathcal H_3\}_g + \\
    + & \{\{\mathcal H_2, \mathcal H_3\}_h, \mathcal H_1\}_g + \{\{\mathcal H_3, \mathcal H_1\}_h, \mathcal H_2\}_g = \\
    = & \int \limits_{0}^{2 \pi} \Bigg(\pd{g^{\alpha \beta}}{u^q} h^{q \gamma} \mathcal a_{\alpha} \mathcal b^{(1)}_{\beta} \mathcal c^{(k)}_{\gamma} - \pd{\Gamma^{\alpha \beta}_s}{u^q} h^{q\gamma} u^s_x \mathcal a_{\alpha} \mathcal b_{\beta} \mathcal c^{(k)}_{\gamma} - \Gamma^{\alpha \beta}_{q} h^{q\gamma} \mathcal a_{\alpha} \mathcal b_{\beta} \mathcal c^{(k + 1)}_{\gamma}\Bigg) \ddd x + \text{c.p} = \\
    = & \int \limits_{0}^{2 \pi} \Bigg( \Gamma^{\alpha \beta}_{q} h^{q\gamma} \mathcal a^{(1)}_{\alpha} \mathcal b_{\beta} \mathcal c^{(k)}_{\gamma} - \Gamma^{\beta \alpha}_{q} h^{q\gamma} \mathcal a_{\alpha} \mathcal b^{(1)}_{\beta} \mathcal c^{(k)}_{\gamma} + \Big( \pd{\Gamma^{\alpha \beta}_{q}}{u^s}h^{q \gamma} - \pd{\Gamma^{\alpha \beta}_{s}}{u^q}h^{q\gamma}\Big) u^s_{x} \mathcal a_{\alpha} \mathcal b_{\beta} \mathcal c^{(k)}_{\gamma} \Bigg) \ddd x + \text{c.p.} 
    \end{aligned}
\end{equation}
Here c.p. stands for cyclic permutations in $\mathcal a_{\alpha}, \mathcal b_{\beta}, \mathcal c_{\gamma}$ and we used identity \eqref{g1}. We denote
\begin{equation}
    A^{\alpha \beta \gamma} = \Gamma^{\alpha \beta}_q h^{q \gamma} \quad \text{   and   } \quad T^{\alpha \beta \gamma} = \Bigg( \pd{\Gamma^{\alpha \beta}_{q}}{u^s}h^{q \gamma} - \pd{\Gamma^{\alpha \beta}_{s}}{u^q}h^{q\gamma}\Bigg) u^s_{x}.
\end{equation}
The compatibility condition \eqref{g3} in these notations takes form
\begin{equation}\label{g4}
    \begin{aligned}
    0 = & \int \limits^{2 \pi}_{0} A^{\alpha \beta \gamma} \mathcal a^{(1)}_{\alpha} \mathcal b_{\beta} \mathcal c^{(k)}_{\gamma} \ddd x - \int \limits^{2 \pi}_{0} A^{\beta \alpha \gamma} \mathcal a_{\alpha} \mathcal b^{(1)}_{\beta} \mathcal c^{(k)}_{\gamma} \ddd x  + \int \limits^{2 \pi}_{0} A^{\beta \gamma \alpha} \mathcal a^{(k)}_{\alpha} \mathcal b^{(1)}_{\beta} \mathcal c_{\gamma} \ddd x - \\
    - & \int \limits^{2 \pi}_{0} A^{\gamma \beta \alpha} \mathcal a^{(k)}_{\gamma} \mathcal b_{\beta} \mathcal c^{(1)}_{\gamma} \ddd x  + \int \limits^{2 \pi}_{0} A^{\gamma \alpha \beta} \mathcal a_{\alpha} \mathcal b^{(k)}_{\beta} \mathcal c^{(1)}_{\gamma} \ddd x - \int \limits^{2 \pi}_{0} A^{\alpha \gamma \beta} \mathcal a^{(1)}_{\alpha} \mathcal b^{(k)}_{\beta} \mathcal c_{\gamma} \ddd x + \\
    + & \int \limits^{2 \pi}_{0} T^{\alpha \beta \gamma} \mathcal a_{\alpha} \mathcal b_{\beta} \mathcal c^{(k)}_{\gamma} \ddd x + \int \limits^{2 \pi}_{0} T^{\beta \gamma \alpha} \mathcal a^{(k)}_{\alpha} \mathcal b_{\beta} \mathcal c_{\gamma} \ddd x + \int \limits^{2 \pi}_{0} T^{\gamma \alpha \beta} \mathcal a_{\alpha} \mathcal b^{(k)}_{\beta} \mathcal c_{\gamma} \ddd x. 
    \end{aligned}
\end{equation}
The following identity holds for arbitrary differential polynomials $\mathcal h_1, \mathcal h_2$
$$
0 = \int \limits_{0}^{2 \pi} D (\mathcal h_1 \mathcal h_2) \ddd x = \int \limits_{0}^{2 \pi}D (\mathcal h_1) \mathcal h_2 \ddd x + \int \limits_{0}^{2 \pi} \mathcal h_1 D(\mathcal h_2) \ddd x.
$$
It is usually referred to as "integration by parts". Applying this identity to the terms in \eqref{g4} we get
$$
\begin{aligned}
& \int \limits^{2 \pi}_{0} A^{\alpha \beta \gamma} \mathcal a^{(1)}_{\alpha} \mathcal b_{\beta} \mathcal c^{(k)}_{\gamma} \ddd x = - \sum \limits_{i + j + r = k} \binom{k}{i, j, r} \int \limits_{0}^{2 \pi} D^{i} \big(A^{\alpha \beta \gamma}\big) \mathcal a^{(1 + j)}_{\alpha} \mathcal b^{(r)}_{\beta} \mathcal c_{\gamma} \ddd x, \\
& \int \limits^{2 \pi}_{0} A^{\beta \alpha \gamma} \mathcal a_{\alpha} \mathcal b^{(1)}_{\beta} \mathcal c^{(k)}_{\gamma} \ddd x = - \sum \limits_{i + j + r = k} \binom{k}{i, j, r} \int \limits^{2 \pi}_{0} D^{i} \big(A^{\beta \alpha \gamma}\big) \mathcal a^{(j)}_{\alpha} \mathcal b^{(r + 1)}_{\beta} \mathcal c_{\gamma} \ddd x, \\
& \int \limits^{2 \pi}_{0} T^{\alpha \beta \gamma} \mathcal a_{\alpha} \mathcal b_{\beta} \mathcal c^{(k)}_{\gamma} \ddd x  = - \int \limits^{2 \pi}_{0} \sum \limits_{i + j + r = k} \binom{k}{i, j, r} D^i \big( T^{\alpha \beta \gamma}\big) \mathcal a^{(j)}_{\alpha} \mathcal b^{(r)}_{\beta} \mathcal c_{\gamma} \ddd x,
\end{aligned}
$$
$$
\begin{aligned}
& \int \limits^{2 \pi}_{0} A^{\gamma \beta \alpha} \mathcal a^{(k)}_{\gamma} \mathcal b_{\beta} \mathcal c^{(1)}_{\gamma} \ddd x = - \int \limits^{2 \pi}_{0} D(A^{\gamma \beta \alpha}) \mathcal a^{(k)}_{\gamma} \mathcal b_{\beta} \mathcal c_{\gamma} \ddd x - \int \limits^{2 \pi}_{0} A^{\gamma \beta \alpha} \mathcal a^{(k + 1)}_{\gamma} \mathcal b_{\beta} \mathcal c^{(1)}_{\gamma} \ddd x - \int \limits^{2 \pi}_{0} A^{\gamma \beta \alpha} \mathcal a^{(k)}_{\gamma} \mathcal b^{(1)}_{\beta} \mathcal c^{(1)}_{\gamma} \ddd x, \\
& \int \limits^{2 \pi}_{0} A^{\gamma \alpha \beta} \mathcal a_{\alpha} \mathcal b^{(k)}_{\beta} \mathcal c^{(1)}_{\gamma} \ddd x = - \int \limits^{2 \pi}_{0} D(A^{\gamma \alpha \beta}) \mathcal a_{\alpha} \mathcal b^{(k)}_{\beta} \mathcal c_{\gamma} \ddd x - \int \limits^{2 \pi}_{0} A^{\gamma \alpha \beta} \mathcal a^{(1)}_{\alpha} \mathcal b^{(k)}_{\beta} \mathcal c_{\gamma} \ddd x - \int \limits^{2 \pi}_{0} A^{\gamma \alpha \beta} \mathcal a_{\alpha} \mathcal b^{(k + 1)}_{\beta} \mathcal c_{\gamma} \ddd x.
\end{aligned}
$$
Here $\binom{n}{p, q, r} = \frac{n!}{p! q! r!}$ are trinomial coefficients. We substitute these equations into \eqref{g4}. After denoting the coefficients in front of the monomial $\mathcal a^{(i)}_{\alpha} \mathcal b^{(j)}_{\beta} \mathcal c_{\gamma}$ as $M^{\alpha \beta \gamma}_{ij}$, the condition \eqref{g4} takes form
\begin{equation}\label{g5}
0 = \int \limits_{0}^{ 2 \pi} \sum \limits_{i, j} M^{\alpha \beta \gamma}_{ij} \mathcal a^{(i)}_\alpha \mathcal b^{(j)}_\beta \mathcal c_{\gamma} \ddd x.    
\end{equation}
Now recall, that $\mathcal a_{\alpha}, \mathcal b_{\beta}, \mathcal c_{\gamma}$ are arbitrary functions in $x$. We perform a rather standard analytical trick: fix $x_0 \in \mathbb R$, $\alpha, \beta, \gamma \leq n$ and $i, j$ with $i + j \leq k + 1$. We take $\mathcal a_q = 0$ for $q \neq \alpha$ and $\mathcal b_r = 0$ for $r \neq \beta$. On $\mathcal a_{\alpha}$ and $\mathcal b_{\beta}$ we put conditions $\mathcal a^{(p)}_\alpha (x_0) = \delta^p_i$ and $\mathcal b^{(q)}_\alpha (x_0) = \delta^q_j$.

We take $\mathcal c_q = 0$ for $q \neq \gamma$ and consider the sequence of smooth functions $\mathcal c^n_{\gamma} (x)$ which converge to delta-function $\delta(x - x_0)$ for $n \to \infty$. We get, that \eqref{g5} converges to $M^{\alpha \beta \gamma}_{ij} (x_0) = 0$. By the construction it holds for all functions $u^\alpha$. As $x_0$ was chosen arbitrary, we get, that \eqref{g2} is equivalent to the collection of the conditions
\begin{equation}\label{g6}
    M^{\alpha \beta \gamma}_{ij} = 0, \quad 1 \leq \alpha, \beta, \gamma \leq n, \quad 0 \leq i + j \leq k + 1.
\end{equation}
Now we are going to study this system.

\begin{Lemma}\label{lemm1}
The compatibility condition \eqref{g4} implies:
\begin{enumerate}
    \item $T^{\alpha \beta \gamma} \equiv 0$
    \item $A^{\alpha \beta \gamma}$ are constants and totally symmetric in $\alpha, \beta, \gamma$ 
\end{enumerate}
\end{Lemma}
\begin{proof}
Fix arbitrary $\alpha, \beta, \gamma$. From \eqref{g6} and \eqref{g4} we get 
$$
M_{00}^{\alpha \beta \gamma} = - D^k \big(T^{\alpha \beta \gamma}\big) = 0.
$$
Recall that $T^{\alpha \beta \gamma}$ is a homogeneous differential polynomial of degree one, with coefficients that do not depend on $x$. By Lemma \ref{lemm5} the vanishing of the derivative implies that this polynomial is identically zero.

Now consider the following coefficients:
\begin{equation}\label{calc4}
\begin{aligned}
    & 0 = M^{\alpha \beta \gamma}_{k + 1 \,0} =  A^{\gamma \beta \alpha} - A^{\alpha \beta \gamma}, \\
    & 0 = M^{\alpha \beta \gamma}_{k\, 1} = - k A^{\alpha \beta \gamma} + A^{\beta \alpha \gamma} + A^{\gamma \beta \alpha} + A^{\beta \gamma \alpha}, \\
    & 0 = M^{\alpha \beta \gamma}_{1\, k} = - A^{\alpha \beta \gamma} + k A^{\beta \alpha \gamma} - A^{\gamma \alpha \beta} - A^{\alpha \gamma \beta}, \\
    & 0 = M^{\alpha \beta \gamma}_{k\, 0} = - k D\big[A^{\alpha \beta \gamma}\big] + D\big[A^{\gamma \beta \alpha}\big] + T^{\beta \gamma \alpha} - T^{\alpha \beta \gamma}.
\end{aligned}
\end{equation}
The first equation from \eqref{calc4} yields $A^{\gamma \beta \alpha} = A^{\alpha \beta \gamma}$. Note, this holds for all $\alpha, \beta, \gamma$. The second and third equations together give 
\begin{equation*}
    \begin{aligned}
    0 & = M^{\alpha \beta \gamma}_{k\, 1} + M^{\alpha \beta \gamma}_{1\, k} = \\
    & = - k A^{\alpha \beta \gamma} + A^{\beta \alpha \gamma} + A^{\gamma \beta \alpha} + A^{\beta \gamma \alpha} - A^{\alpha \beta \gamma} + k A^{\beta \alpha \gamma} - A^{\gamma \alpha \beta} - A^{\alpha \gamma \beta} = \\
    & = (k + 1) \big( A^{\beta \alpha \gamma} - A^{\alpha \beta \gamma}\big)
    \end{aligned}
\end{equation*}
Here we used $A^{\beta \gamma \alpha} = A^{\alpha \gamma \beta}$, which follows from the eailier as $\alpha, \beta, \gamma$ are arbitrary.

So, we have shown, that $A^{\alpha \beta \gamma} = A^{\beta \alpha \gamma}$ and $A^{\alpha \beta \gamma} = A^{\gamma \beta \alpha}$. The permutation group in three elements is generated by arbitrary two transpositions, thus, the corresponding object is totally symmetric in upper indices.

Finally, the last equation of \eqref{calc4} yields $(1 - k) D \big(A^{\alpha \beta \gamma}\big) = 0$. As $A^{\alpha \beta \gamma}$ is a differential polynomial of degree zero, which does not depend on $x$, we get that it is constant.
\end{proof}

Due to the symmetry of $A^{\alpha \beta \gamma}$, the second equation in \eqref{calc4} takes form
$$
0 = (3 - k) A^{\alpha \beta \gamma}.
$$
Thus, for $k \geq 5$ we get that $A^{\alpha \beta \gamma}$ identically vanish. Recall, that $A^{\alpha \beta \gamma} = \Gamma^{\alpha \beta}_s h^{s \gamma} = g^{\alpha q} \Gamma^\beta_{qs} h^{s \gamma}$. As both $h$ and $g$ are non-degenerate, we get, that $\Gamma^\beta_{qs}$ identically vanish and in coordinates $u^1, \dots, u^n$ both metrics are in constant form.

We have shown, that the statement of the Lemma \ref{lemm1} follows from conditions \eqref{g6}, which, in turn, are equivalent to the compatibility of $\mathcal A_g$ and $\mathcal B_h$. Let us show, that in case $k = 3$ the statement of the Lemma \ref{lemm1} is, in fact, equivalent to the compatibility. The compatibility condition in the form \eqref{g4} is written as
\begin{equation*}
\begin{aligned}
    0 = & A^{\alpha \beta \gamma} \Bigg( \int \limits^{2 \pi}_{0} \mathcal a^{(1)}_{\alpha} \mathcal b_{\beta} \mathcal c^{(3)}_{\gamma} \ddd x - \int \limits^{2 \pi}_{0} \mathcal a_{\alpha} \mathcal b^{(1)}_{\beta} \mathcal c^{(3)}_{\gamma} \ddd x  + \int \limits^{2 \pi}_{0} \mathcal a^{(3)}_{\alpha} \mathcal b^{(1)}_{\beta} \mathcal c_{\gamma} \ddd x - \\
    - & \int \limits^{2 \pi}_{0} \mathcal a^{(3)}_{\gamma} \mathcal b_{\beta} \mathcal c^{(1)}_{\gamma} \ddd x  + \int \limits^{2 \pi}_{0}  \mathcal a_{\alpha} \mathcal b^{(3)}_{\beta} \mathcal c^{(1)}_{\gamma} \ddd x - \int \limits^{2 \pi}_{0}  \mathcal a^{(1)}_{\alpha} \mathcal b^{(3)}_{\beta} \mathcal c_{\gamma} \ddd x \Bigg) = \\
    = & A^{\alpha \beta \gamma} \Bigg(- \int \limits^{2 \pi}_{0} \mathcal a^{(4)}_{\alpha} \mathcal b_{\beta} \mathcal c_{\gamma} \ddd x - 3 \int \limits^{2 \pi}_{0} \mathcal a^{(3)}_{\alpha} \mathcal b^{(1)}_{\beta} \mathcal c_{\gamma} \ddd x - 3 \int \limits^{2 \pi}_{0} \mathcal a^{(2)}_{\alpha} \mathcal b^{(2)}_{\beta} \mathcal c_{\gamma} \ddd x - \int \limits^{2 \pi}_{0} \mathcal a^{(1)}_{\alpha} \mathcal b^{(3)}_{\beta} \mathcal c_{\gamma} \ddd x + \\
    & + \int \limits^{2 \pi}_{0} \mathcal a^{(3)}_{\alpha} \mathcal b^{(1)}_{\beta} \mathcal c_{\gamma} \ddd x + 3 \int \limits^{2 \pi}_{0} \mathcal a^{(2)}_{\alpha} \mathcal b^{(2)}_{\beta} \mathcal c_{\gamma} \ddd x + 3\int \limits^{2 \pi}_{0} \mathcal a^{(1)}_{\alpha} \mathcal b^{(3)}_{\beta} \mathcal c_{\gamma} \ddd x + \int \limits^{2 \pi}_{0} \mathcal a_{\alpha} \mathcal b^{(4)}_{\beta} \mathcal c_{\gamma} \ddd x + \\
    & + \int \limits^{2 \pi}_{0} \mathcal a^{(4)}_{\alpha} \mathcal b_{\beta} \mathcal c_{\gamma} \ddd x + \int \limits^{2 \pi}_{0} \mathcal a^{(3)}_{\alpha} \mathcal b^{(1)}_{\beta} \mathcal c_{\gamma} \ddd x - \int \limits^{2 \pi}_{0} \mathcal a^{(1)}_{\alpha} \mathcal b^{(3)}_{\beta} \mathcal c_{\gamma} \ddd x - \int \limits^{2 \pi}_{0} \mathcal a_{\alpha} \mathcal b^{(4)}_{\beta} \mathcal c_{\gamma} \ddd x + \\
    & + \int \limits^{2 \pi}_{0} \mathcal a^{(3)}_{\alpha} \mathcal b^{(1)}_{\beta} \mathcal c_{\gamma} \ddd x - \int \limits^{2 \pi}_{0}  \mathcal a^{(1)}_{\alpha} \mathcal b^{(3)}_{\beta} \mathcal c_{\gamma} \ddd x \Bigg) = 0.
\end{aligned}
\end{equation*}
Again, we recall, that $A^{\alpha \beta \gamma} = \Gamma^{\alpha \beta}_s h^{s \gamma}$ and $h^{s \gamma}$ is in constant form. Thus, $\Gamma^{\alpha \beta}_s$ are constant. As $A^{\alpha \beta \gamma} = A^{\beta \alpha \gamma}$ we get, that $\Gamma^{\alpha \beta}_s$ are also symmetric in upper indices.

We denote $a^{\alpha \beta}_s = - \Gamma^{\alpha \beta}_s$. The condition \eqref{g1} on contrvariant Christoffel symbols implies that $g^{\alpha \beta} = 2 (b^{\alpha \beta} + a^{\alpha \beta}_s u^s)$. Here we denoted the constant term as $2 b^{\alpha \beta}$, this is exactly the matrix of $g^{\alpha \beta}$ at the coordinate origin, thus, it is non-degenerate. The second condition \eqref{g2} yields relation $a^{\alpha \beta}_q (b^{q \gamma} + a^{q\gamma}_s u^s) = a^{\gamma \beta}_q (b^{q \alpha} + a^{q\alpha}_s u^s)$. This yields two conditions 
\begin{equation}\label{fr1}
a^{\alpha \beta}_q b^{q \gamma} = a^{\gamma \beta}_q b^{q \alpha}, \quad a^{\alpha \beta}_q a^{q \gamma}_s = a^{\gamma \beta}_q a^{q \alpha}_s.    
\end{equation}
Finally, the symmetry $A^{\alpha \beta \gamma} = A^{\gamma \beta \alpha}$ yields $a^{\alpha \beta}_q h^{q\gamma} = a^{\gamma \beta}_q h^{q\alpha}$. Thus, we have shown that the conditions of Theorem \ref{t2} are necessary.

For $k \geq 5$ it is obvious, that these conditions are also sufficient. For $k = 3$ the sufficiency can be found in the literature (see \cite{balnov},\cite{str}), but to keep our work self-contained we prove this fact. 

We say that coordinates $u^1, \dots, u^n$ are Frobenius coordinates for contrvariant metric $g^{\alpha \beta}$ if
\begin{equation}\label{frob}
g^{\alpha \beta} = 2(b^{\alpha \beta} + a^{\alpha \beta}_s u^s),    
\end{equation}
where constants $b^{\alpha \beta}, a^{\alpha \beta}_s$ satisfy \eqref{fr1}. Note that by construction, the constants are both symmetric in upper indices. Recall that the name comes from the fact that $a^{\alpha \beta}_s$ define commutative associative algebra with Frobenius form $b^{\alpha \beta}$. The next Lemma proves a stronger result, then the one we need.

\begin{Lemma}\label{lemm2}
Let $g$ be a contrvariant metric and $u^1, \dots, u^n$ a coordinate system. The following two conditions are equivalent:
\begin{enumerate}
    \item In coordinates $u^1, \dots, u^n$ the contrvariant Christoffel symbols of $g$ are constant and symmetric in upper indices
    \item The metric $g^{\alpha \beta}$ is in Frobenius coordinates.
\end{enumerate}
If either of these conditions hold, the metric is flat and $- \Gamma^{\alpha \beta}_s = a^{\alpha \beta}_s$.
\end{Lemma}
\begin{proof}
We have already proven that in the statement of the Lemma condition 1 implies condition 2. Now, let us show the inverse. 

In given coordinates define Christoffel symbols of the Levi-Civita connection of $g^{\alpha \beta}$ as
$$
\Gamma^\beta_{rs} = - g_{rq} a^{q\beta}_s.
$$
By construction $\nabla_s g^{\alpha \beta} = 2a^{\alpha \beta}_s - a^{\alpha \beta}_s - a^{\beta \alpha}_s = 0$. At the same time conditions \eqref{fr1} imply, that $a^{\alpha \beta}_s g^{s\gamma} = g^{\alpha r} \Gamma^{\beta}_{rs} g^{s \gamma} = g^{\alpha r} \Gamma^{\beta}_{sr} g^{s \gamma}$. Thus, the connection is symmetric and $g$ is parallel along it. Thus, it implies that the corresponding connection is Levi-Civita connection and $a^{\alpha \beta}_s = - \Gamma^{\alpha \beta}_s$. Thus, condition 2 of Lemma implies condition 1.

Finally, we apply the general formula for Riemann tensor of $g$ in terms of contrvariant Christoffel symbols (formula 0.9 in \cite{dub1}):
\begin{equation*}
R^{\beta \gamma \alpha}_s = - \pd{\Gamma^{\gamma \alpha}_p}{u^s} g^{p \beta} + \pd{\Gamma^{\gamma \alpha}_s}{u^p} g^{p \beta} + \Gamma^{\beta \gamma}_q \Gamma^{q \alpha}_s - \Gamma^{\beta \alpha}_q \Gamma^{q \gamma}_s
\end{equation*}
The contrvariant Christoffel symbols are constants. We get $R^{\beta \gamma \alpha}_s = a^{\beta \gamma}_q a^{q \alpha}_s - a^{\beta \alpha}_q a^{q \gamma}_s = a^{\gamma \beta}_q a^{q \alpha}_s - a^{\alpha \beta}_q a^{q \gamma}_s = 0$. The last condition is the second condition from \eqref{fr1} and here we used symmetry of $a^{\alpha \beta}_s$ in upper indices. The Lemma is proved. 
\end{proof}

The Lemma \ref{lemm2} implies that the conditions of Theorem \ref{t2} for $k = 3$ are sufficient. Now, let us prove Corollary \ref{cor1}.

By construction $S^{\alpha \beta}_s = g^{\alpha q} \big( \Gamma^\beta_{qs} - \bar \Gamma^\beta_{qs}\big)$. Fix Darboux coordinates $u^1, \dots, u^n$ of $\mathcal B_h$. We get that in given coordinates $\Gamma^{\alpha \beta}_s = S^{\alpha \beta}_s = - a^{\alpha \beta}_s$. This implies that the conditions of Corollary \ref{cor1} are equivalent to those of Theorem \ref{t2} in Darboux coordinates. 

Both Theorem \ref{t2} and Corollary \ref{cor1} are proved.


\section{Proof of Theorem \ref{t4}} \label{proof3}

We start with trivial case: the classification of pairs of bilinear forms over $\mathbb R$ can be found in \cite{lan} (Theorem 9.2). Each pair $h^{\alpha \beta}, b^{\alpha \beta}$ of non-degenerate forms in dimension two can be brought to one of the following form
\begin{equation*}
    \begin{aligned}
    & \left\{\left(\begin{array}{cc}
         h_1 & 0  \\
         0 & h_2 
    \end{array}\right), \left(\begin{array}{cc}
         1 & 0  \\
         0 & 1 
    \end{array}\right)\right\},  \left\{\left(\begin{array}{cc}
         h_1 & 0  \\
         0 & h_2 
    \end{array}\right), \left(\begin{array}{cc}
         - 1 & 0  \\
         0 & 1 
    \end{array}\right)\right\}, \left\{\left(\begin{array}{cc}
         h_1 & 0  \\
         0 & h_2 
    \end{array}\right), \left(\begin{array}{cc}
         - 1 & 0  \\
         0 & - 1 
    \end{array}\right)\right\},\\
    & \left\{\left(\begin{array}{cc}
         0 & h_1  \\
         h_1 & 1 
    \end{array}\right), \left(\begin{array}{cc}
         0 & 1  \\
         1 & 0 
    \end{array}\right)\right\}, \left\{\left(\begin{array}{cc}
         - h_2 & h_1  \\
         h_1 & h_2 
    \end{array}\right), \left(\begin{array}{cc}
         0 & 1  \\
         1 & 0 
    \end{array}\right)\right\}\\
    \end{aligned}
\end{equation*}
In each family, different collections of constants $h_1, h_2$ provide pairs, which cannot be transformed one into the other by coordinate change. This provides families $(01)$ through $(05)$ in the statement of Theorem \ref{t4}.

We will need the following reformulation of \eqref{i1}. In a given basis, consider a pair of operators $R_1 \eta = \eta \star \eta_1$ and $R_2 \eta = \eta \star \eta_2$. Bilinear form $h^{\alpha \beta}$ with matrix $H$ is invariant form if and only if both operators are self-adjoint with respect to this form, that is
$$
R_1^TH = H R_1, \quad R_2^TH = H R_2.
$$

We proceed as follows: for each algebra from the list \eqref{alg} we explicitly write $R_1, R_2$ and $h^{\alpha \beta}$, calculate using the above observation (we omit the calculations itself). 

Then we describe the coordinate changes that preserve the structure relation of the algebras from \eqref{alg}. Using these coordinate changes, we bring $h$ to some normal form. In case $\mathfrak a_i, i \geq 2$ the corresponding coordinate system is unique and, thus, different collection of constants $b_j, h_j, j = 1,2$ provide non-equivalent normal forms of inhomogeneous Hamiltonian operators of type $1 + 3$. 

In case of $\mathfrak a_1$ there are two canonical coordinate systems, but the corresponding coordinate changes preserve the form of $b^{\alpha \beta}$, thus, the normal form is unique.

\textbf{Case $\mathfrak a_1$: } In a given basis $\eta^1, \eta^2$ the only non-zero structure relation of this algebra is
$$
\eta^2 \star \eta^2 = \eta^2.
$$
The operators $R_1, R_2$ and invariant form $h^{\alpha \beta}$ are
$$
R_1 = \left(\begin{array}{cc}
     0 & 0  \\
     0 & 0 
\end{array}\right), \quad R_2 = \left(\begin{array}{cc}
     0 & 0  \\
     0 & 1 
\end{array}\right), \quad h^{\alpha \beta} = \left(\begin{array}{cc}
     h_1 & 0  \\
     0 & h_2 
\end{array}\right).
$$
The coordinate change that preserves the structure relations of $\mathfrak a_1$ is
$$
\bar \eta^1 = \alpha \eta^1, \quad \bar \eta^2 = \eta^2.
$$
Under this transformation the coefficient $h_1$ can be made into either $1$ or $-1$. The corresponding coordinate system is not uniquely defined. Besides trivial coordinate change we have one more: $\bar \eta^1 = - \eta^1, \bar \eta^2 = \eta^2$. At the same time this coordinate change does not change the normal form, so we get two normal forms $(11)$ and $(12)$.

\textbf{Case $\mathfrak a_2$:} In a given basis $\eta^1, \eta^2$ the only non-zero structure relation of this algebra is
$$
\eta^2 \star \eta^2 = \eta^1.
$$
The operators $R_1, R_2$ and invariant form $h^{\alpha \beta}$ are
$$
R_1 = \left(\begin{array}{cc}
     0 & 0  \\
     0 & 0 
\end{array}\right), \quad R_2 = \left(\begin{array}{cc}
     0 & 1  \\
     0 & 0 
\end{array}\right), \quad h^{\alpha \beta} = \left(\begin{array}{cc}
     0 & h_1  \\
     h_1 & h_2 
\end{array}\right).
$$
The coordinate change that preserves the structure relation of $\mathfrak a_2$ is
\begin{equation}\label{sr1}
\bar \eta^1 = \alpha^2 \eta^1, \quad \bar \eta^2 = \beta \eta^1 + \alpha \eta^2.    
\end{equation}
Taking coordinate change in the form
$$
\bar \eta^1 = \frac{1}{(h_1)^{2/3}} \eta^1, \quad \bar \eta^2 = - \frac{h_2}{2 (h_1)^{4/3}} \eta^1 + \frac{1}{(h_1)^{1/3}} \eta^2
$$
we bring $h^{\alpha \beta}$ to the form $h_1 = 1, h_2 = 0$. 

Coordinate change \ref{sr1} that preserves $h^{\alpha \beta}$, satisfies conditions $h(\bar \eta^1, \bar \eta^2) = \alpha^3 = 1$ and $h(\bar \eta^2, \bar \eta^2) = 2 \alpha \beta = 0$. This implies that the only such coordinate change is trivial one. Thus, the canonical coordinate system is unique and we obtain the normal form $(21)$.

\textbf{Case $\mathfrak a_3$:} In a given basis $\eta^1, \eta^2$ the non-zero structure relations of this algebra are
$$
\eta^2 \star \eta^2 = \eta^2, \quad \eta^1 \star \eta^2 = \eta^1. 
$$ 
The operators $R_1, R_2$ and invariant form $h^{\alpha \beta}$ are
$$
R_1 = \left(\begin{array}{cc}
     0 & 1  \\
     0 & 0 
\end{array}\right), \quad R_2 = \left(\begin{array}{cc}
     1 & 0  \\
     0 & 1 
\end{array}\right), \quad h^{\alpha \beta} = \left(\begin{array}{cc}
     0 & h_1  \\
     h_1 & h_2 
\end{array}\right).
$$
The coordinate change that preserves the structure relations of $\mathfrak a_3$ is
$$
\bar \eta^1 = \alpha \eta^1, \quad \bar \eta^2 = \eta^2.
$$
We take $\alpha = \frac{1}{h_1}$ and get the invariant form with $h_1 = 1$. Following the same argument as in previous case, we get that the corresponding coordinate system is unique. We get the normal form $(31)$.

\textbf{Case $\mathfrak a_4$:}  In a given basis $\eta^1, \eta^2$ the non-zero structure relations of this algebra are
$$
\eta^1 \star \eta^1 = \eta^1, \quad \eta^2 \star \eta^2 = \eta^2. 
$$ 
In given coordinates operators $R_1, R_2$ and invariant form $h^{\alpha \beta}$ are
$$
R_1 = \left(\begin{array}{cc}
     1 & 0  \\
     0 & 0 
\end{array}\right), \quad R_2 = \left(\begin{array}{cc}
     0 & 0  \\
     0 & 1 
\end{array}\right), \quad h^{\alpha \beta} = \left(\begin{array}{cc}
     h_1 & 0  \\
     0 & h_2 
\end{array}\right).
$$
The only coordinate change that preserves the structure relations of algebra $\mathfrak a_4$ is the trivial one. Thus, the corresponding coordinate system is unique and we obtain the normal form $(41)$.

\textbf{Case $\mathfrak a_5$:} In a given basis $\eta^1, \eta^2$ the non-zero structure relations of this algebra are
$$
\quad \quad \eta^2 \star \eta^2 = \eta^2, \quad \eta^1 \star \eta^1 = - \eta^2, \quad \eta^1 \star \eta^2 = \eta^1.. 
$$
The operators $R_1, R_2$ and invariant form $h^{\alpha \beta}$ are
$$
R_1 = \left(\begin{array}{cc}
     0 & 1  \\
     - 1 & 0 
\end{array}\right), \quad R_2 = \left(\begin{array}{cc}
     1 & 0  \\
     0 & 1 
\end{array}\right), \quad h^{\alpha \beta} = \left(\begin{array}{cc}
     - h_2 & h_1  \\
     h_1 & h_2 
\end{array}\right).
$$
The only coordinate change that preserves the structure relations of algebra $\mathfrak a_5$ is the trivial one. Same as in the previous case, the corresponding coordinate system is unique. We get the normal form $(51)$. The Theorem is proved.


\section{Proof of Theorem \ref{t3}} \label{pr3} \label{proof4}

We start with several algebraic lemmas.

\begin{Lemma}\label{lemm3}
Consider vector space $V$ of dimension $n$ and its dual $V^*$. Fix a basis $\eta^1, \dots, \eta^n$ in $V^*$ and the dual basis in $V$. Assume that in given basis the collection of constants $a^{\alpha \beta}_s$ defines the structure of a commutative associative algebra. Assume also that $b^{\alpha \beta}$ is Frobenius form of this algebra. Then
$$
c^\beta_{rs} = \frac{1}{2} b_{rq} a^{q \beta}_s
$$
defines the structure of a commutative associative algebra on $V$ in the dual basis. In other words
$$
c^\beta_{ps} = c^\beta_{sp}, \quad c^\beta_{sq} c^q_{pr} = c^q_{sp} c^\beta_{qr}.
$$
This algebra is isomorphic to the one on $V^*$. In particular, if $e \in V^*$ is unity of $a^{\alpha \beta}_s$ with coordinates $e_1, \dots, e_n$, then $f \in V$ with coordinates $f^p = 2 b^{pq} e_q$ is unity of $c^\beta_{ps}$.
\end{Lemma}
\begin{proof}
Symmetric non-degenerate bilinear form $2 b^{\alpha \beta}$ can be understood as invertible map $b: V^* \to V$. We get
\begin{equation*}
    \frac{1}{2} b_{sp} b_{rq} a^{pq}_m b^{m \beta} = \frac{1}{2} b_{sp} b_{rq} a^{\beta q}_m b^{mp} = \frac{1}{2} b_{rq} a^{\beta q}_m \delta^m_s = \frac{1}{2} b_{rq} a^{q \beta}_s = c^{\beta}_{rs}.
\end{equation*}
This implies that $c^\beta_{rs}$ is commutative associative algebra, isomorphic to the one defined by $a^{\alpha \beta}_s$ on $V^*$. It also implies the statement about the unity. The Lemma is proved.
\end{proof}

The next fact from linear algebra is well-known, but to keep our work self-contained we provide it with proof.

\begin{Lemma}\label{wwr}
Consider a real vector space $V$ of dimension $n$ and a non-degenerate linear operator $R: V \to V$. Then there exists a polynomial $p(t)$ (possibly with complex coefficients), such that for $L = p(R)$ one has $L^2 = R$.
\end{Lemma}
\begin{proof}
First, assume that $R$ has no Jordan blocks, then one can take $p(t)$ to be the Lagrange polynomial for square roots of eigenvalues. Here the complex coefficients may appear, for example, for negative real eigenvalues.

Now assume that $R$ is in general form. Consider its Jordan-Chevalley decomposition $R = R_s + R_n$, where $R_s, R_n$ are semisimple and nilpotent parts respectively. Moreover, there exist polynomials $q_s(t)$ and $q_n(t)$ both without constant terms, such that $R_s = q_s(R), R_n = q_n(R)$ (Proposition p.17, \cite{hum}). We rewrite it as $R = R_s (\operatorname{Id} + R^{-1}_s R_n)$. We get
$$
N = R^{-1}_s R_n = \Big(q_s(R)\Big)^{-1} q_n(R).
$$
The inverse of a matrix $R^{-1}_s$ can be written as a polynomial, thus, we get that $N$ is a polynomial in $R$.

By earlier argument there exists a good root of $R_s$, which is expressed as a polynomial of $R_s$ and, thus, as a polynomial of $R$. We denote it as $R^{1/2}_s$. 

As $N$ is nilpotent $(\operatorname{Id} + N)^{-1/2}$ is a polynomial in $R$. Finally we take
$$
L = R^{1/2}_s (\operatorname{Id} + N)^{-1/2}.
$$
to be a good root of $R$ in general form. 
\end{proof}

\begin{Lemma}\label{lemm4}
Assume that $c^\beta_{ps}$ and $r^\beta_p$ are defined as in the statement of Theorem \ref{t3}. Consider $l^\beta_q$ to be the components of good square root of $r^\alpha_p$ and $\bar l^\beta_q$ the components of the inverse of the good square root. The identity
\begin{equation}\label{rim}
m^\beta_r c^r_{ps} = c^\beta_{pr} m^r_s = m^r_p c^\beta_{rs}
\end{equation}
holds for both $m^\alpha_p = l^\alpha_p$ and $m^\alpha_p = \bar l^\alpha_p$.
\end{Lemma}
\begin{proof}
First, let us check that \eqref{rim} holds for $m^\alpha_p = r^\alpha_p$. The constants $a^{\alpha \beta}_s, b^{\alpha \beta}, h^{\alpha \beta}$ are the ones defined in the statement of Theorem \ref{t3}. By definition of Frobenius forms we get
\begin{equation*}
    \begin{aligned}
    h^{\alpha s} a_s^{\beta \gamma} = h^{\alpha s} b_{sm} b^{mr} a^{\beta \gamma}_r = 2 r^\alpha_m b^{\gamma r} a^{\beta m}_r
    \end{aligned}
\end{equation*}
Lowering index $\gamma$ with $b$ and renaming the rest of the indices we arrive to the identity
$$
h^{\alpha s} c^\beta_{ps} = 2 r^\alpha_m a^{\beta m}_s.
$$
At the same time, as $h^{\alpha \beta}$ is invariant and $a^{\alpha \beta}_s$ is commutative, we get
$$
2 r^\alpha_m a^{\beta m}_s = h^{\alpha s} c^\beta_{ps} = h^{\beta s} c^{\alpha}_{ps}.
$$
Lowering $\beta$ with $b$, and, again, renaming the indices we get
$$
2 r^{\alpha}_m c^m_{ps} = c^{\alpha}_{pm} r^m_s
$$
We know that $c^\beta_{ps} = c^\beta_{sp}$ (Lemma \ref{lemm3}), so $r^\alpha_p$ satisfies \eqref{rim}.

Now recall that if identity \eqref{rim} holds for $r^\alpha_p$, then it holds for all powers of $r^\alpha_p$ and, thus, for all polynomials of the operator. As $l$ is a good root, then the identity \eqref{rim} holds for $l^\alpha_p$. The inverse of a non-degenerate matrix is a polynomial of this matrix, thus, by the same argument the identity \eqref{rim} holds for $\bar l^\alpha_p$.
\end{proof}

\begin{Lemma}\label{lemm5}
Consider a differential polynomial $\mathcal h$ and assume that $D^j(\mathcal h) = 0$ for some $j \geq 1$. Then $\mathcal h$ is a constant.
\end{Lemma}
\begin{proof}
It is enough to prove this for a homogeneous differential polynomial of degree $m$. Let us order the terms with respect to "algebraic degree", that is the number of symbols $u^\alpha_{x^j}$ for $j \geq 1$. 

We denote the decomposition of $\mathcal h$ into the sum of homogeneous (with respect to this ordering) terms as $\mathcal h = \mathcal h_{(i)} + \dots $, where $1 \leq i \leq m$ and $\mathcal h_{(i)} \neq 0$. By construction $D (\mathcal h_{(i)}) = \bar{\mathcal h}_{(i)} + \bar{\mathcal h}_{(i + 1)}$ and functional coefficients of $\bar{\mathcal h}_{(i)}$ coincide with functional coefficients of $\mathcal h_{(i)}$.

Applying $j$ times $D$ we get that vanishing of the result implies vanishing of the term of order $(i)$. This, in turn, means, that $\mathcal h_{(i)} = 0$. This contradiction completes the proof. 
\end{proof}

Now let us proceed to the proof of the Theorem.

First, let us discuss the solvability of \eqref{rq2}. The constants $c^\beta_{ps}$ are symmetric in lower indices (Lemma \ref{lemm3}). From \eqref{rq2}, the formulas for the components $\mathcal v_i^\alpha$ of series $\mathrm v^\alpha$ are
\begin{equation}\label{inv}
\begin{aligned}
u^\beta & = \mathcal v_1^{\beta} + \frac{1}{2} c^\beta_{ps} \mathcal v_1^p \mathcal v_1^s, \\
\Big(\delta^\beta_q + c^\beta_{qs} \mathcal v^s_1 \Big) \mathcal v^q_i & = - \frac{1}{2} \sum \limits_{j = 2}^{i - 1} c^\beta_{ps} \mathcal v_j^p \mathcal v_{i - j}^s + l^\beta_s (\mathcal v^s_{i - 1})_x, \quad i \geq 2
\end{aligned}    
\end{equation}
The first formula is actually Balinskii-Novikov formula from \cite{balnov}. It is written in the form $u^\alpha = F^\alpha(\mathcal v^1, \dots, \mathcal v^n)$. At the point $u^\alpha = 0, \mathcal v^\beta_1 = 0$ we have $\pd{F^\alpha}{\mathcal v^q} = \delta^\alpha_q$, thus, by the implicit function theorem in a neighbourhood of this point the solutions $\mathcal v^\alpha_1(u)$ exist.

For $i \geq 2$ for sufficiently small values of $\mathcal v_1^\alpha$ the operator $\delta^\beta_q + \frac{1}{2} c^\beta_{qs} \mathcal v^s_1$ on the l.h.s. of \eqref{inv} is invertible. Thus, the solution exists and each series is uniquely defined modulo the choice of solution of Balinskii-Novikov formula.

Fix two discs $\mathrm M^n$ and $\mathrm N^n$ with coordinates $u^1, \dots, u^n$ and $\mathrm v^1, \dots, \mathrm v^n$. We denote the r.h.s. of formula \ref{rq2} as $\phi (\mathrm v) = (\phi^1(\mathrm v), \dots, \phi^n(\mathrm v))$. For arbitrary element $c \in \Omega \mathrm N^n$ we have $\phi \circ c : \mathbb R \to \mathrm M^n$. That is $\phi$ defines a map $\phi: \Omega \mathrm N^n \to \Omega \mathrm M^n$. This induces the dual map $\phi^*$ on the functionals: for a given density $\mathcal h (u, u_x, \dots)$ we have
$$
\phi^* \mathcal H = \int \limits^{2 \pi}_{0} \mathcal h(\phi(\mathrm v), \phi_x(\mathrm v), \dots) \ddd x.
$$
The density $\mathcal h(\phi(\mathrm v))$ is not a differential polynomial in general. It has natural description: it is polynomial in $\mathrm v^\alpha_{x^j}, j = 2, \dots$ with coefficients being a smooth functions of $2n$ variables $\mathrm v^\alpha, \mathrm v^\alpha_x$. So in this section the polynomiality is not assumed.

The straightforward calculation of variational derivative $\vpd{(\phi^* \mathcal H)}{\mathrm v^s}$ yields: for arbitrary $r \in \Omega \mathrm N^n$ we have
$$
\begin{aligned}
& \pd{}{\epsilon} \int \limits^{2 \pi}_{0} \mathcal h\big(u(\mathrm v + \epsilon \mathrm r, \mathrm v_x + \epsilon \mathrm r_x)\big) \ddd x  = \int \limits^{2 \pi}_{0} \sum \limits_{j = 0} \Bigg( \pd{\mathcal h}{u^q_{x^j}} D^j\Big(\pd{u^q}{\mathrm v^s} r^s \Big) + \pd{\mathcal h}{u^q_{x^j}} D^j\Big(\pd{u^q}{\mathrm v_x^s} r_x^s \Big)\Bigg) \ddd x + \dots  = \\
= & \int \limits^{2 \pi}_{0} \sum \limits_{j = 0} \Bigg( \pd{\mathcal h}{u^q_{x^j}} D^j\Big(\big(\delta^q_s +   c^q_{ps} \mathrm v^p\big) r^s \Big) + \pd{\mathcal h}{u^q_{x^j}} D^{j + 1} (l^q_s r^s) \Bigg)\ddd x + \dots  = \\
= & \int \limits^{2 \pi}_{0} \Bigg( \vpd{\mathcal h}{u^q} \big(\delta^q_s + c^q_{ps} \mathrm v^p\big) r^s + l^q_s \vpd{\mathcal h}{u^q} r^s_x \Bigg) \ddd x + \dots = \int \limits^{2 \pi}_{0} \Bigg( \vpd{\mathcal h}{u^q} \big(\delta^q_s + c^q_{ps} \mathrm v^p\big) - l^q_sD\Big(\vpd{\mathcal h}{u^q}\Big) \Bigg) r^s \ddd x + \dots 
\end{aligned}
$$
The calculation obviously holds for the densities $\mathrm v = \mathcal v_1 + \mathcal v_2 + \dots$, where $\mathcal v_i$ is either zero ore differential polynomial of degree $i - 1$ in $u^\alpha_{x^j}, j \geq 0$. For functional $\mathrm V^\alpha = \int^{2 \pi}_{0} \mathrm v^\alpha \ddd x$, where $\mathrm v^\alpha$ is one of the solutions of \eqref{rq2}, we get identity
\begin{equation}\label{prot1}
    \vpd{\mathrm V^\alpha}{\mathrm v^s} = \delta^\alpha_s = \Big(\delta^q_s + c^q_{ps} \mathrm v^p\Big) \vpd{\mathrm V^\alpha}{u^q} - l^q_s D \Big( \vpd{\mathrm V^\alpha}{u^q}\Big).
\end{equation}
Recall, that identity \eqref{rq2} provides a recursion formula for calculation of the components $\mathcal v_i$ of the series $\mathrm v = \mathcal v_1 + \mathcal v_2 + \mathcal v_3 + \dots$. The identity \eqref{prot1}, in turn, provides the recursion formulas for calculation of the components $\vpd{\mathcal v^\alpha_i}{u^s}$ of the series $\vpd{\mathrm v^\alpha}{u^s} = \vpd{\mathcal v^\alpha_1}{u^s} + \vpd{\mathcal v^\alpha_2}{u^s} + \dots$. 

Now we are going to show that together these two identities imply the first statement of the theorem by a straightforward computation. Recall that $\bar l^\beta_s$ stands for the inverse of the operator $l^\beta_s$. We rewrite \eqref{rq2} and \eqref{prot1} as
\begin{equation}\label{prot2}
\begin{aligned}
& D\Big( \vpd{\mathrm V^\alpha}{u^s}\Big) =  \Big(\bar l^m_s + \bar l^r_p c^m_{pr} \mathrm v^p\Big) \vpd{\mathrm V^\alpha}{u^m} - \bar l^\alpha_s, \\
& \mathrm v^\alpha_x = \bar l^\alpha_m u^m - \bar l^\alpha_m \mathrm v^m - \frac{1}{2} \bar l^\alpha_m c^m_{ps} \mathrm v^p \mathrm v^s.
\end{aligned}
\end{equation}
Applying $D$ to both sides of \eqref{prot1} we get
\begin{equation*}
    0 = c^q_{ps} \mathrm v^p_x \vpd{\mathrm V^\alpha}{u^q} + \Big(\delta^q_s + c^q_{ps} \mathrm v^p\Big) D \Big(\vpd{\mathrm V^\alpha}{u^q}\Big) - l^m_s D^2 \Big( \vpd{\mathrm V^\alpha}{u^m}\Big).
\end{equation*}
Substituting \eqref{prot2} into previous equation yields
\begin{equation*}
    0 = c^e_{sq} \bar l^q_m \Big(u^m - \mathrm v^m - \frac{1}{2}c^m_{pr} \mathrm v^p \mathrm v^r\Big) \vpd{\mathrm V^\alpha}{u^e} + \Big(\delta^q_s + c^q_{ps} \mathrm v^p\Big) \bar l^m_q\Big(\delta^e_m + c^e_{rm} \mathrm v^r\Big) \vpd{\mathrm V^\alpha}{u^e} - \bar l^\alpha_s - \bar l^\alpha_m c^m_{ps} \mathrm v^p - l^m_s D^2 \Big( \vpd{\mathrm V^\alpha}{u^m}\Big).
\end{equation*}
We multiply both sides by $l^\alpha_q$ and apply Lemma \ref{lemm4}. This yields
\begin{equation*}
0 = c^q_{sp} u^p \vpd{\mathrm V^\alpha}{u^q} + \Big(\delta^m_s + c^m_{ps} \mathrm v^p + c^q_{ps} c^m_{qr} \mathrm v^p \mathrm v^r - \frac{1}{2} c^m_{sq} c^q_{pr} \mathrm v^p \mathrm v^r \Big) \vpd{\mathrm V^\alpha}{u^m} - \delta^\alpha_s - c^\alpha_{ps} \mathrm v^p - r^m_s D^2 \Big( \vpd{\mathrm V^\alpha}{u^m}\Big).
\end{equation*}
Using identity $c^q_{ps} c_{qr}^m = c_{sq}^m c^q_{pr}$ from Lemma \ref{lemm3}, we finally get
$$
0 = c^q_{sp} u^p \vpd{\mathrm V^\alpha}{u^q} + \vpd{\mathrm V^\alpha}{u^s} + \Big( c^m_{ps} \mathrm v^p + \frac{1}{2} c^m_{sq} c^q_{pr} \mathrm v^p \mathrm v^r \Big) \vpd{\mathrm V^\alpha}{u^m} - \delta^\alpha_s - c^\alpha_{ps} \mathrm v^p - r^m_s D^2 \Big( \vpd{\mathrm V^\alpha}{u^m}\Big)
$$
Now we again apply $D$ to both sides. This yields
$$
\begin{aligned}
    0 & = c^q_{sp} u^p_x \vpd{\mathrm V^\alpha}{u^q} + c^q_{sp} u^p D \Big(\vpd{\mathrm V^\alpha}{u^q}\Big) + D\Big(\vpd{\mathrm V^\alpha}{u^s}\Big) + \Big( c^m_{ps} \mathrm v_x^p + c^m_{sq} c^q_{pr} \mathrm v_x^p \mathrm v^r \Big) \vpd{\mathrm V^\alpha}{u^m} + \\
    & + \Big( c^m_{ps} \mathrm v^p + \frac{1}{2} c^m_{sq} c^q_{pr} \mathrm v^p \mathrm v^r \Big) D\Big(\vpd{\mathrm V^\alpha}{u^m}\Big) - c^\alpha_{ps} \mathrm v^p_x - r^m_s D^3 \Big( \vpd{\mathrm V^\alpha}{u^m}\Big)
\end{aligned}
$$
Here we have again used identities for $c^\beta_{ps}$ from Lemma \ref{lemm3}. Regrouping terms we get
$$
\begin{aligned}
0 & = c^q_{sp} u^p_x \vpd{\mathrm V^\alpha}{u^q} + c^q_{sp} u^p D \Big(\vpd{\mathrm V^\alpha}{u^q}\Big) + D\Big(\vpd{\mathrm V^\alpha}{u^s}\Big) + c^r_{sp} \mathrm v^p_x \Bigg( \Big( \delta^m_r + c^m_{rq} \mathrm v^q \Big) \vpd{\mathrm V^\alpha}{u^m} - \delta^\alpha_r \Bigg) + \\
    & + \Big( c^m_{ps} \mathrm v^p + \frac{1}{2} c^m_{sq} c^q_{pr} \mathrm v^p \mathrm v^r \Big) D\Big(\vpd{\mathrm V^\alpha}{u^m}\Big) - r^m_s D^3 \Big( \vpd{\mathrm V^\alpha}{u^m}\Big)
\end{aligned}
$$
From the first formula of \eqref{prot2} we get that the expression in big brackets equals to $l_s^m D\big( \vpd{\mathrm V^\alpha}{u^m}\big)$. Regrouping coefficients and applying \eqref{rq2} we get
$$
0 = c^q_{sp} u^p_x \vpd{\mathrm V^\alpha}{u^q} + 2 c^q_{sp} u^p D \Big(\vpd{\mathrm V^\alpha}{u^q}\Big) + D\Big(\vpd{\mathrm V^\alpha}{u^s}\Big) - r^m_s D^3 \Big( \vpd{\mathrm V^\alpha}{u^m}\Big).
$$
Raising the indices by $2b^{\alpha q}$ we get 
\begin{equation}\label{last}
2 ( b^{\alpha q} +  a^{\alpha q}_s u^s) D \Big(\vpd{\mathrm V^\alpha}{u^q}\Big) + a^{\alpha \beta}_s u^s_x \vpd{\mathrm V^\alpha}{u^q} = h^{\alpha q} D^3 \Big( \vpd{\mathrm V^\alpha}{u^q}\Big)    
\end{equation}
Now recall that if $\mathcal v_i$ is a differential polynomial of degree $i - 1$, then its partial variational derivative $\vpd{\mathcal v_i}{u^\alpha}$ is also a differential polynomial of degree $i - 1$. At the same time operator $D$ increases the degree of differential polynomial by one. This observation implies that after substituting $\mathrm v^\alpha = \mathcal v^\alpha_1 + \mathcal v^\alpha_2 + \dots$ into \eqref{last} we get two Harry Dym hierarchies
$$
\begin{array}{l}
     \mathcal A^{\beta q} \vpd{\mathcal V^\alpha_1}{u^q} = 0,  \\
     \mathcal A^{\beta q} \vpd{\mathcal V^\alpha_{2i + 1}}{u^q} = \mathcal B^{\beta q} \vpd{\mathcal V^\alpha_{2i + 1}}{u^q}
\end{array} \quad \text{  and  } \quad \begin{array}{l}
     \mathcal A^{\beta q} \vpd{\mathcal V^\alpha_2}{u^q} = 0,  \\
     \mathcal A^{\beta q} \vpd{\mathcal V^\alpha_{2i + 2}}{u^q} = \mathcal B^{\beta q} \vpd{\mathcal V^\alpha_{2i + 1}}{u^q}
\end{array} \quad \text{i = 1, \dots }.
$$
Here $\mathcal V^\alpha_i = \int^{2 \pi}_{0} \mathcal v_i \ddd x$. After renaming $\mathcal h_i = \mathcal v_{2i - 1}$ for $i = 1, \dots$ the first hierarchy produces the first statement of Theorem \ref{t3}. We will need the following Lemma.

As for the second hierarchy taking flat coordinates for $\mathcal A_g$, we get that by Lemma \ref{lemm5} the $\vpd{\mathcal V_i}{u^q} = 0$ for all even $i$. Thus, the second hierarchy is trivial and $\mathcal v_i$ for even $i$ are total derivatives.

Recall that $\mathcal A_g$ and $\mathcal B_h$ are compatible Poisson brackets. The constructed solutions provide the Lenard-Magri chains for corresponding bracket. Thus, the second statement of Theorem \ref{t3} follows from general bihamiltonian formalism for such brackets.

Now assume that we start Harry Dym hierarchy with $\mathcal V^\alpha_1 = 0$. The second equation of the hierarchy becomes
$$
\mathcal B^{\beta q} D^3\Big( \vpd{\mathcal H^\alpha_2}{u^q}\Big) = 0.
$$
By Lemma \ref{lemm5} this implies that in the case of homogeneous differential polynomial densities, the only solution is $\mathcal H^\alpha_2 = 0$. Applying the same argument, we get that the entire hierarchy is zero. Thus, the third statement if proved.

To prove the last statement of the theorem we first recall that by Lemma \ref{lemm3}, the algebra $c^\beta_{ps}$ has unity $f = (f^1, \dots, f^n)$. Adding $\frac{1}{2} f^\alpha$ to both sides of \eqref{rq2} we get
$$
u^\alpha + \frac{1}{2} f^\alpha = \frac{1}{2} f^\alpha + \mathrm v^\alpha + \frac{1}{2} c^\beta_{ps} \mathrm v^p \mathrm v^s + l^\beta_s \mathrm v^s_x = \frac{1}{2} c^\beta_{ps} (\mathrm v^p + f^p) (\mathrm v^s + f^s) + l^\beta_s (\mathrm v^s + f^s)_x.
$$
At the same time we get the following identity
$$
\frac{1}{2} a^{\beta \alpha}_q f^q = b^{\alpha m} c^\beta_{mq} f^q = b^{\alpha \beta}.
$$
Thus, the formula for $\mathcal A^{\alpha q}_g$ takes form
$$
\mathcal A^{\beta q}_g = 2 (b^{\beta q} + a^{\beta q}_s u^s) D + a^{\beta q}_s u^s_x = 2 a^{\beta q}_s (u^s + \frac{1}{2} f^s) D + a^{\beta q}_s (u^s + \frac{1}{2} f^s)_x.
$$
Shifting $u^\beta \to u^\beta + \frac{1}{2} f^\beta$ we get the fourth statement of the Theorem.


\section{Appendix: some facts about commutative associative algebras with Frobenius forms}

In this section we collect some basic facts about commutative associative algebras and their invariant forms. We always assume that algebra is over $\mathbb R$ or $\mathbb C$. If the field is not specified, the result holds for both.

\begin{Proposition}\label{pp1}
Consider a commutative associative algebra $\mathfrak a$ and arbitrary $m \in \mathfrak a^*$. Then $h_m (\xi, \eta) = m(\xi \star \eta)$ is an invariant form. We call this form exact.
\end{Proposition}
\begin{proof}
We have
$$
h_m(\xi \star \eta, \zeta) = m((\xi \star \eta)\star \zeta) = m(\xi \star (\eta \star \zeta)) = h_m(\xi, \eta \star \zeta).
$$
The Proposition is proved.
\end{proof}

Among others we deal with commutative associative algebras with unity

\begin{Proposition}\label{pp2}
Consider a commutative associative algebra $\mathfrak a$. Then
\begin{enumerate}
    \item If $\mathfrak a$ has unity, then every invariant form is exact
    \item If there exists a non-degenerate exact invariant form, then $\mathfrak a$ has unity
\end{enumerate}
\end{Proposition}
\begin{proof}
First, assume that $h$ is an invariant form for commutative associative algebra $\mathfrak a$ with unity $e$. Define $m \in \mathfrak a^*$ as $m(\eta) = h(e, \eta)$. We have
$$
h(\xi, \eta) = h(e, \xi \star \eta) = h_m (\xi, \eta).
$$
The first part of Proposition is proved. Now proceed to the second part. Assume that $h_m$ is non-degenerate exact invariant form for $m \in \mathfrak a^*$. It can be treated as map $h_m: \mathfrak a \to \mathfrak a^*$. We write 
\begin{equation}\label{mm1}
h_m (\xi, \zeta) = m (\xi \star \zeta) = h_m (\xi \star \zeta, h_m^{-1}m) = h_m(\xi \star h_m^{-1}m, \zeta).    
\end{equation}
Denote $\eta = h_m^{-1}m$ and consider the operator $R_{\eta} \xi = \xi \star \eta$. The formula \eqref{mm1} takes form
$$
h_m(\xi, \eta) = h_m(R_{\eta} \xi, \zeta)
$$
By construction $R_{\eta} = h_m^{-1} h_m = \operatorname{Id}$ and $\eta$ is unity. The Proposition is proved.
\end{proof}

The commutative associative algebra $\mathfrak a$ is called reducible if it can be written as $\mathfrak a = \mathfrak a' \oplus \mathfrak a''$ for some commutative associative algebras $\mathfrak a', \mathfrak a''$, and irreducible otherwise. By definition, every commutative associative algebra is a sum of irreducible algebras.

\begin{Proposition}\label{pp3}
Assume that a commutative associative algebra $\mathfrak a$ is reducible, that is
\begin{equation}\label{dec}
    \mathfrak a = \mathfrak a' \oplus \mathfrak a'',
\end{equation}
where $\mathfrak a', \mathfrak a''$ are both commutative associative algebras. If $\mathfrak a$ has unity and invariant non-degenerate form $h$, then both $\mathfrak a', \mathfrak a''$ have unities, decomposition \eqref{dec} is orthogonal with respect to $h$ and restrictions of $h$ on each summand are non-degenerate invariant forms.
\end{Proposition}
\begin{proof}
Denote unity in $\mathfrak a$ as $e$. Consider the decomposition $e = e' + e''$ of the unity with respect to \eqref{dec}. By construction both $e', e''$ are unities for $\mathfrak a', \mathfrak a''$ respectively.

For arbitrary $\xi \in \mathfrak a, \eta \in \mathfrak a''$ we have $\xi \star \eta = 0$ and, thus, 
$$
h(\xi, \eta) = h(\xi, \eta \star e) = h(\xi \star \eta, e) = 0.
$$
Finally, if either of $h', h''$ has kernel, then this is the kernel of the entire $h$. The proposition is proved.
\end{proof}

\end{document}